\pdfoutput=1
\documentclass{amsart}

\usepackage{hyperref} % formerly supplied by eptcs.cls
\hypersetup{
  hidelinks,
  pdftitle={Essential Unitarity for Higher-Order Quantum Computation},
  pdfauthor={Samson Abramsky and Radha Jagadeesan},
  pdfsubject={Categorical semantics for higher-order quantum computation},
  pdfkeywords={essential unitarity, categorical quantum mechanics,
    Kelly--Laplaza, coherent quantum control, quantum switch}
}

\usepackage{amsmath,amssymb,amsthm}
\usepackage{mathrsfs}  % \mathscr notation
\usepackage{stmaryrd}
\let\parr\bindnasrepma
\usepackage{tikz}
\usepackage{tikz-cd}
\usetikzlibrary{positioning, calc}
\usepackage{enumitem}
\setlist{nosep,leftmargin=*}% tight vertical spacing and flush margins for all lists
\usepackage{mathtools}
\usepackage{comment}

\newif\ifconfversion
\confversionfalse% ← change to \confversiontrue for the conference version

\ifconfversion
  \AtBeginDocument{%
    \setlength{\abovedisplayskip}{2pt plus 1pt minus 1pt}%
    \setlength{\belowdisplayskip}{2pt plus 1pt minus 1pt}%
    \setlength{\abovedisplayshortskip}{0pt plus 1pt minus 1pt}%
    \setlength{\belowdisplayshortskip}{0pt plus 1pt minus 1pt}%
    \setlength{\topsep}{2pt}%
    \setlength{\partopsep}{0pt}%
  }
\fi

\newif\ifshowproofs
\showproofstrue% ← change to \showproofsfalse to hide proofs in extended mode

\ifconfversion
  \showproofsfalse
\fi

\ifshowproofs\else
  \excludecomment{proof}
\fi

\newcommand{\extendedonly}[1]{\ifconfversion\else#1\fi}
\newcommand{\confonly}[1]{\ifconfversion#1\fi}

\usepackage{iftex}
\ifpdf
  \usepackage{underscore}
  \usepackage[T1]{fontenc}
  \usepackage{lmodern}
\else
  \usepackage{breakurl}
\fi

\newtheorem{definition}{Definition}[section]
\newtheorem{theorem}[definition]{Theorem}
\newtheorem{lemma}[definition]{Lemma}
\newtheorem{proposition}[definition]{Proposition}
\newtheorem{corollary}[definition]{Corollary}

\theoremstyle{definition}
\newtheorem{remark}[definition]{Remark}

\numberwithin{equation}{section}

\newcommand{\Perm}{\mathrm{Perm}}

\newcommand{\QC}{\mathbb{QC}}

\newcommand{\id}{\mathrm{id}}

\DeclareMathOperator{\Pos}{Pos}
\DeclareMathOperator{\Neg}{Neg}

\DeclareMathOperator{\ev}{ev}

\newcommand{\CC}{\mathbb{C}}
\newcommand{\NN}{\mathbb{N}}
\newcommand{\RR}{\mathbb{R}}

\newcommand{\ii}{\mathrm{i}}
\newcommand{\ket}[1]{\lvert #1\rangle}
\newcommand{\bra}[1]{\langle #1\rvert}
\newcommand{\Unit}{\mathbf{1}}

\title[Essential Unitarity for Higher-Order Quantum Computation]
      {Essential Unitarity for Higher-Order Quantum Computation}

\author[Abramsky]{Samson Abramsky}
\address{University College London, London, UK}
\email{s.abramsky@ucl.ac.uk}

\author[Jagadeesan]{Radha Jagadeesan}
\address{DePaul University, Chicago, USA}
\email{rjagadee@depaul.edu}

\begin{document}

%============================================================
% amsart requires the abstract BEFORE \maketitle.
%============================================================
\begin{abstract}
\confonly{\emph{(Extended abstract.)}\;\,}We develop a boundary-centric
semantic framework for higher-order quantum computation, building on
the Kelly--Laplaza description of compact closure and Abramsky's
execution account. In the semantic carrier \(\Perm(\CC)\), morphisms
are complex-linear combinations of polarized boundary linkings, composed
by execution. Finite-family addresses provide coherent control over
finite-level quantum registers (qudits) while retaining the multiplicative
boundary structure.

We identify \emph{essential unitarity}, a boundary condition extending
ordinary unitarity to higher-order interfaces. On positive qudit
registers it coincides with ordinary matrix unitarity; at higher order
it expresses preservation of information across the full polarized
boundary.

We define a unit-free coherent quantum core generated by multiplicative
wiring, unitary gates on positive qudit registers, and contextual
coherent control, and prove that every one of its morphisms is
essentially unitary. The framework realizes an applied coherent quantum switch
and the unitary stages of equal-ratio one-slot supermap dilations with
explicit memory.

\extendedonly{An extended abstract of this work was accepted for QPL 2026
and is forthcoming in its proceedings.}
\end{abstract}

\maketitle

\section{Introduction}
\label{sec:intro}

Since the foundational work of Abramsky and
Coecke~\cite{abramskyCoecke2004}, Categorical Quantum Mechanics has
shown that tensorial composition, duality, and entanglement fit within
dagger compact closed categories~\cite{Selinger2007CPM}. Adding a
symmetric monoidal sum accommodates classical interfaces and
measurement~\cite{Coecke_Kissinger_2017}. String diagrams give a
representation-free account of process composition, independent of
any Hilbert-space model.

In finite-dimensional closed-system quantum mechanics, reversible
dynamics is represented by unitary operators: an evolution $f:A\to A$
satisfies $f^\dagger f=ff^\dagger=\id_A$. Higher-order processes,
however, take processes rather than states as inputs. The quantum
switch~\cite{chiribella2013quantum}, for instance, coherently
superposes the two causal orders $f\circ g$ and $g\circ f$ of
evolutions $f,g:A\to A$ under quantum control; it is physically
realizable and reversible, yet is not a first-order unitary
endomorphism because its inputs are processes. This raises the
question: \emph{What boundary invariant extends first-order unitarity
to higher-order typed interfaces?}

We work from below. The base is the free compact closed category of
Kelly--Laplaza (KL)~\cite{coherence}, in the concrete combinatorial
presentation of Abramsky~\cite{AbramskyAS}, which separates
permutations, loops, and polarities. Its finite-coproduct completion
$\mathrm{Cop}$ records distinguishable layouts. Rather than mediate
distributivity by coherence isomorphisms~\cite{Laplaza72}, we work in
distributed normal form: objects are finite disjoint unions of
interfaces, and tensor distributes over sum at the object level. The
$\CC$-linear matrix/biproduct completion is the carrier category $\Perm(\CC)$,
a dagger compact closed category with finite biproducts, supplying
superposition and quantum amplitudes
(Section~\ref{sec:structural-core}).

A KL morphism is read on the two polarities of its compact-name
boundary. A basis linking contributes its permutation matrix; linear
combinations contribute the corresponding complex matrix, with closed
loop factors already in the coefficients. Essential unitarity (EU)
requires the resulting boundary matrix to be two-sided unitary
(Definition~\ref{def:EU}). On positive qudit registers this is ordinary
matrix unitarity. For qudits carrying multiplicative payloads, the full
matrix retains both the qudit-family address and the KL port
(Section~\ref{sec:eu-predicate}).

The source category $\QC$ (Definition~\ref{def:basic-QC}) has raw
types generated from positive qudit literals and their duals by
$\otimes$ and $\parr$. Its morphisms are generated by the unit-free no-Mix fragment of
multiplicative linear logic (MLL), unitary gates on positive qudit registers,
and contextual qudit control. A controlled family may share an
arbitrary higher-order input, but has a common positive output and
retains one selector port throughout. The two carrier counterexamples
of Section~\ref{subsec:QC-why-coherence} explain these restrictions.

For a source type \(X\), let \(\mathsf A_X^+\) and
\(\mathsf A_X^-\) collect its positive and negative boundary
coordinates, respectively. For \(h:A\to B\), compact naming exposes
the boundary \((A^*,B)\). Its fixed polarization determines the
incoming and outgoing axes, and typed un-naming gives
\[
 M(h):
 \mathsf A_A^+\otimes\mathsf A_B^-
 \longrightarrow
 \mathsf A_A^-\otimes\mathsf A_B^+.
\]
Our main theorem proves both Gram identities for \(M(h)\). Its
fixed-width representation is therefore a unitary numerical
qudit--MLL matrix, hence \(h\) is essentially unitary
(Theorem~\ref{thm:calculus-EU}).

The proof follows the source syntax. An administrative representation
makes every constructor matrix explicit. Compact coherence and the
naturality of contextual control replace composition through a mixed
intermediate type by composition through pure, meaning one-polarity,
intermediate types. The
unit-free multiplicative base, positive-register gates, contextual
control, and composition through pure interfaces all preserve the two
Gram identities. This yields a direct proof path from a source term to
its essentially unitary full matrix.

The calculus also contains the applied contextual quantum switch of
Section~\ref{subsec:contextual-switch-v2}. Function types support
resource-linear one-slot transformations by pre- and postcomposition,
tensoring with constants, and synchronized qudit control. Finally, the
calculus realizes the unitary stages of deterministic one-slot
supermap dilations whenever the input/output dimension ratios agree
and the memory remains an explicit boundary wire
(Section~\ref{sec:expressiveness}). The recovery identity is stated in
external Hilbert-space/channel semantics; measurement, partial trace,
and mixed-state semantics are left to future work.

\subsection{Related Work}\label{sec:related-work}

\noindent\emph{Categorical quantum mechanics.}
Diagrammatic calculi such as the ZX-calculus~\cite{ZX} provide rewrite
systems for linear maps on qubits; dagger Frobenius algebras capture
observables and classical interfaces~\cite{Coecke2013}; and the CPM
construction~\cite{Selinger2007CPM} shows how mixed-state quantum
mechanics and completely positive maps arise canonically from dagger
compact structure.

\noindent\emph{Higher-order quantum processes.}
The study of \emph{supermaps}---transformations on quantum
operations---originates
with~\cite{Chiribella2008Supermaps,Chiribella2008,Chiribella2009},
which characterize deterministic one-slot supermaps and give
sequential realizations by pre- and post-isometries connected by a
memory system, followed by partial trace over the final memory.
Subsequent work studies causal and indefinite-order
processes~\cite{Oreshkov2012,A2014,KissingerUijlen2019,WilsonChiribellaKissingerLocality}
and higher-order-process accounts based, respectively, on BV logic,
closed symmetric monoidal categories, enriched monoidal categories, and
strong profunctors~\cite{simmons2022higherorder,Wilson2021HOPT,Wilson2022Framework,HeffordWilson2024}.
These approaches characterize higher-order admissibility relative to
an underlying first-order process theory through operational, causal,
or structural conditions. Two concurrent language-level accounts place
higher-order quantum control in causal process semantics. Hirata and
Tsukada~\cite{HirataTsukada2026} develop a pure, measurement-free
higher-order calculus with quantum conditionals and a causality type
system, interpreted in $\mathbf{CausHilb}$, which connects
finite-dimensional Hilbert spaces with $\mathbf{Caus}[\mathbf{CPM}]$.
Barsse, P\'echoux, and Perdrix~\cite{BarssePechouxPerdrix2026}
extend coherent control to measurement and arbitrary quantum channels,
with denotational semantics in $\mathbf{Caus}[\mathbf{CPM}]$. At the
level of physical operations, $\QC$ isolates the qudit-controlled,
essentially unitary boundary fragment of this shared design space, leaving
measurement and discarding external.

\noindent\emph{The present approach.}
The construction begins with
Kelly--Laplaza compact-closed syntax and its complex-linear
finite-biproduct carrier. The source category $\QC$ adds positive
qudit-register gates and contextual qudit control to the unit-free
multiplicative base. Correctness is expressed by essential unitarity of
the full qudit--MLL boundary matrix of each $\QC$-morphism.

Our resource-linear contexts provide a unitary analogue of part of the
usual supermap
picture~\cite{WilsonChiribellaKissingerLocality,HeffordWilson2024}:
the distinguished hole is used once, substitution remains in $\QC$,
and its ambient carrier extension is $\CC$-linear. Wilson observes that symmetric polycategories arise by forgetting
representability from linearly distributive categories and support the
partial applicability of holes to bipartite
processes~\cite{Wilson2026HigherOrderCircuits}. The unit-free
multiplicative structure of $\QC$ determines a symmetric polycategory
representable on nonempty lists; its weak distributors implement
polycategorical composition on nonempty
interfaces~\cite{CockettSeely1997WDC}. Resource-linear contexts provide
a source-level analogue of
the single-hole local-application fragment isolated by Wilson and
Chiribella~\cite{WilsonChiribella2026MonoidalSU}. Equal input/output
dimension ratios permit padding the ancillary spaces and extending the
pre- and post-isometries to unitaries while memory wires remain
explicit. An internal source/carrier treatment of measurement, partial
trace, and mixed-state semantics is left to future work.

\noindent\emph{Structure of the paper.}
Section~\ref{sec:orientation} introduces the boundary viewpoint and
essential unitarity.
Section~\ref{sec:structural-core} constructs the finite-coproduct
precursor $\mathrm{Cop}$ and the carrier $\Perm(\CC)$.
Sections~\ref{sec:boundary-transport} and~\ref{sec:eu-predicate}
develop KL boundary matrices, qudit--MLL matrices, and EU.
Sections~\ref{sec:QC} and~\ref{sec:calculus} define $\QC$ and prove
that every $\QC$-morphism is essentially unitary.
Section~\ref{sec:expressiveness} develops resource-linear contexts
and equal-ratio one-slot dilation stages with explicit memory.
\extendedonly{The appendices provide the carrier, boundary-execution,
and motivating-witness calculations used in the main text.}

\section{Essential unitarity, informally}
\label{sec:orientation}

%----------------------------------------------------------
% Figure 1: KL view and boundary view
%----------------------------------------------------------
\begin{figure}[htbp]
\centering
\begin{tikzpicture}[scale=0.85,
  port/.style={font=\scriptsize, inner sep=1pt},
  wire/.style={semithick},
]
  % === Panel (a): KL view ===
  \begin{scope}[local bounding box=klview]
    \node[port] (s1) at (0,1.8)   {$c^{+}$};
    \node[port] (s2) at (0,1.2)   {$a_1^{+}$};
    \node[port] (s3) at (0,0.6)   {$a_2^{+}$};
    \node[port] (s4) at (0,0)     {$b^{-}$};
    \node[port] (t1) at (2.6,1.8) {$b_1^{+}$};
    \node[port] (t2) at (2.6,1.2) {$b_2^{+}$};
    \node[port] (t3) at (2.6,0.6) {$b_3^{+}$};
    \node[port] (t4) at (2.6,0)   {$a^{-}$};
    \draw[wire] (s1) -- (t1);
    \draw[wire] (s2) -- (t2);
    \draw[wire] (s3) -- (t3);
    \draw[wire] (s4) -- (t4);
  \end{scope}
  \node[below=3pt of klview.south, font=\scriptsize]
    {(a)\ KL view};

  % === Panel (b): Boundary view ===
  \begin{scope}[xshift=4.4cm, local bounding box=bview]
    \node[port] (p1) at (0,1.8)   {$a^{-}$};
    \node[port] (p2) at (0,1.2)   {$b_1^{+}$};
    \node[port] (p3) at (0,0.6)   {$b_2^{+}$};
    \node[port] (p4) at (0,0)     {$b_3^{+}$};
    \node[port] (n1) at (2.6,1.8) {$b^{-}$};
    \node[port] (n2) at (2.6,1.2) {$c^{+}$};
    \node[port] (n3) at (2.6,0.6) {$a_1^{+}$};
    \node[port] (n4) at (2.6,0)   {$a_2^{+}$};
    \draw[wire] (p1) -- (n1);
    \draw[wire] (p2) -- (n2);
    \draw[wire] (p3) -- (n3);
    \draw[wire] (p4) -- (n4);
  \end{scope}
  \node[below=3pt of bview.south, font=\scriptsize]
    {(b)\ boundary view};
\end{tikzpicture}

\caption{A KL matching and its polarized boundary view.}
\label{fig:kl-vs-boundary}
\end{figure}
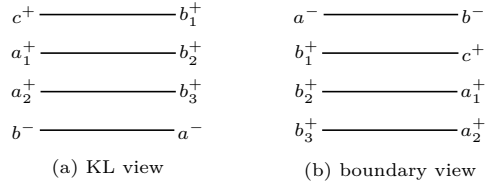
The starting point is the Kelly--Laplaza picture of a
morphism~\cite{coherence}, in the abstract-scalar presentation of
Abramsky~\cite{AbramskyAS}. A KL interface is a signed pair
$A=(A^+,A^-)$, with $\Pos(A):=A^+$ and $\Neg(A):=A^-$. Duality swaps
the two components, while tensor takes their disjoint union componentwise.
A KL morphism
$f:A\to B$ is a port matching
$A^+\sqcup B^-\cong B^+\sqcup A^-$, possibly with closed internal
loops. Its geometry is a wiring diagram.

Take $C=(1,0)$, $A=(2,1)$, and $B=(3,1)$. Then
$C\otimes A=(3,1)$, so a morphism $C\otimes A\to B$ matches four ports
on each side (Figure~\ref{fig:kl-vs-boundary}(a)). A KL basis morphism
is a port matching. The carrier $\Perm(\CC)$ adds complex amplitudes
and finite-family indices; within the single-interface block considered
here, its morphisms are finite $\CC$-linear combinations of such
matchings.

The boundary view regroups the same external ports as
$P_{A,B}:=\Neg(A)\sqcup\Pos(B)=\Pos(A^*\otimes B)$ and
$N_{A,B}:=\Pos(A)\sqcup\Neg(B)=\Neg(A^*\otimes B)$. The KL matching is
then read from $N_{A,B}$ (incoming) to $P_{A,B}$ (outgoing); in the running
example, $|P_{C\otimes A,B}|=|N_{C\otimes A,B}|=4$
(Figure~\ref{fig:kl-vs-boundary}(b)). The compact-name object
$A^*\otimes B$ changes the boundary partition, not the wiring. For KL
interfaces $A,B$, the carrier evaluation
$\ev_{A,B}:(A^*\otimes B)\otimes A\to B$ contracts the $A^*$- and
$A$-ports and carries the $B$-ports through.
%----------------------------------------------------------
% Figure 2: Eval — KL view (A = B = (1,1))
%----------------------------------------------------------
\begin{figure}[ht]
\centering
\begin{tikzpicture}[
  port/.style={font=\scriptsize, inner sep=1pt},
  wire/.style={semithick},
]
  \node[port] (s1) at (0,1.5)   {$\widetilde{a}^{+}$};
  \node[port] (s2) at (0,1.0)   {$b_1^{+}$};
  \node[port] (s3) at (0,0.5)   {$a^{+}$};
  \node[port] (s4) at (0,0)     {$b_2^{-}$};
  \node[port] (t1) at (2.2,1.5) {$a^{-}$};
  \node[port] (t2) at (2.2,1.0) {$b_2^{+}$};
  \node[port] (t3) at (2.2,0.5) {$\widetilde{a}^{-}$};
  \node[port] (t4) at (2.2,0)   {$b_1^{-}$};
  \draw[wire] (s1) -- (t1);
  \draw[wire] (s2) -- (t2);
  \draw[wire] (s3) -- (t3);
  \draw[wire] (s4) -- (t4);
\end{tikzpicture}
\caption{$\ev_{A,B}$.}
\label{fig:eval}
\end{figure}
Evaluation is syntactically higher-order but diagrammatically ordinary:
the $A^*$-part contracts with $A$, and the $B$-part passes through. For
$A=B=(1,1)$ it is the four-wire KL matching of Figure~\ref{fig:eval};
tildes mark ports from the $A^*$-factor, and subscripts distinguish the
two occurrences of $B$. Its higher-order character lies in the typed
boundary partition, not the wiring.
%----------------------------------------------------------
% Tiny currying figure — stacked u / ũ
%----------------------------------------------------------
\begin{figure}[ht]
\centering
\begin{tikzpicture}[
  ty/.style={draw, dashed, rounded corners, inner sep=2pt,
             minimum width=0.45cm, minimum height=0.3cm,
             font=\scriptsize},
  body/.style={draw, semithick, fill=black!4, rounded corners,
               inner sep=3pt, minimum width=0.75cm,
               minimum height=0.9cm, font=\scriptsize},
]
  % tiny panel (a): u — on top
  \begin{scope}[local bounding box=tinya]
    \node[ty]   (Ca) at (-0.85, 0.24) {$C$};
    \node[ty]   (Aa) at (-0.85,-0.24) {$A$};
    \node[body] (ua) at (0,0)         {$u$};
    \node[ty]   (Ba) at ( 0.85, 0)    {$B$};
    \draw[semithick] (Ca.east) -- ($(ua.west)+(0, 0.2)$);
    \draw[semithick] (Aa.east) -- ($(ua.west)+(0,-0.2)$);
    \draw[semithick] (ua.east) -- (Ba.west);
  \end{scope}

  % tiny panel (b): ũ — below
  \begin{scope}[xshift=3.0cm, local bounding box=tinyb]
    \node[ty]   (Cb)  at (-0.85, 0)     {$C$};
    \node[body] (ub)  at ( 0, 0)        {$\tilde u$};
    \node[ty]   (Asb) at ( 0.85, 0.24)  {$A^{*}$};
    \node[ty]   (Bb)  at ( 0.85,-0.24)  {$B$};
    \draw[semithick] (Cb.east) -- (ub.west);
    \draw[semithick] ($(ub.east)+(0, 0.2)$) -- (Asb.west);
    \draw[semithick] ($(ub.east)+(0,-0.2)$) -- (Bb.west);
  \end{scope}
\end{tikzpicture}
\caption{Currying}
\label{fig:currying}
\end{figure}
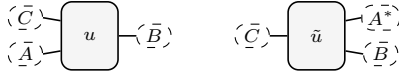
Currying preserves each typed Kelly--Laplaza component and its
coefficient under the canonical identification of compact-name ports
(Figure~\ref{fig:currying} and Remark~\ref{rem:curry-invisible}). It
changes the boundary partition, not the underlying matching.
Let $M,N$ be signed KL interfaces. For a loop-free matching
$g:M\to N$, write $[g]$ for its carrier basis vector,
$\tau_g:N_{M,N}\xrightarrow{\cong}P_{M,N}$ for its boundary bijection,
and $\underline{\tau_g}$ for the corresponding permutation matrix. For
a finite set $S$, let $\CC[S]$ be the Hilbert space with orthonormal
basis $S$. Thus
\[
 F=\sum_g c_g[g]:M\longrightarrow N
\]
has boundary matrix
\[
 \mathsf B_{M,N}(F)
 :=\sum_g c_g\,\underline{\tau_g}:
 \CC[N_{M,N}]\longrightarrow\CC[P_{M,N}],
\]
where the coefficients already contain the evaluated closed-loop
factors. Essential unitarity requires this operator to be two-sided
unitary. It is a boundary condition, not the assertion
that $F$ is a dagger-unitary morphism of $\Perm(\CC)$; evaluation is
the basic example separating the two notions.

Finite-family indices introduce one further coordinate. For a positive
qudit register carrying an MLL payload, the full matrix keeps the qudit
address separate from the KL-port matching. In a fixed-width sector,
every payload row has the same number \(r\) of KL ports; \(i,j\) are
qudit addresses and \(g\) ranges over \(r\)-port matchings. The matrix
has the form
\[
 \sum_{j,i,g}c_{ji,g}\,
 \lvert j\rangle\!\langle i\rvert
 \otimes\underline{\tau_g}.
\]
Section~\ref{sec:eu-predicate} develops this qudit--MLL matrix and its
sequent-level version. For a source morphism $h:A\to B$, compact naming
of its carrier image exposes the boundary $(A^*,B)$. The type fixes the
incoming factors $\mathsf A_A^+\otimes\mathsf A_B^-$ and the outgoing
factors $\mathsf A_A^-\otimes\mathsf A_B^+$; un-naming along this
partition produces $M(h)$. Its fixed-width representation turns the
carrier Gram identities into ordinary unitarity of the full numerical
matrix $\mathbf M(h)$.

Propositions~\ref{prop:omg-prime-h-witness}
and~\ref{prop:QC-F2-witness} motivate the qudit-control discipline of
Section~\ref{sec:QC}. Section~\ref{sec:calculus} proves it essentially
unitary by reducing mixed-intermediate composition to pure-intermediate
composition.\par\smallskip
\emph{Throughout, the test is the same: retain both the qudit address
and the polarized KL ports, form the full boundary matrix, and ask
whether it is unitary.}

\section{The Polarized Carrier Category \texorpdfstring{\(\Perm(\CC)\)}{Perm(C)}}
\label{sec:structural-core}

Three names are reserved throughout: \emph{source} is $\QC$,
\emph{carrier} is $\Perm(\CC)$, \emph{precursor} is $\mathrm{Cop}$.
Put $\Bbbk:=\CC[\delta]$, with $\delta$ a formal indeterminate
recording closed execution cycles.  We construct the polynomial
matrix/biproduct envelope $\Perm(\Bbbk)$ in three steps --- the Kelly--Laplaza
base $\mathbf{KL}$~\cite{AbramskyAS,coherence} (ports), its free
finite-coproduct completion $\mathrm{Cop}$ (finite-family indexing,
with the groupoid $\mathbf{FinBij}$ of finite sets and bijections
governing structural reindexings), and the
full $\Bbbk$-linear matrix/biproduct completion (amplitudes) --- and
then specialize to
the carrier $\Perm(\CC)$, on which essential unitarity will live,
with loop-free KL matrix units as its basis morphisms.
\emph{Notation:} $\mathbf{KL}$ and $\Perm(\CC)$ carry a single
tensor; when a formula tracks source polarity we write
$X\parr Y:=X\otimes Y$ and $f\parr g:=f\otimes g$, the
identification fixed by the De Morgan coherences.

%-----------------------------------------------------------------------------
\subsection{Kelly--Laplaza base and polarized interfaces}
\label{subsec:kl-polarized}
%-----------------------------------------------------------------------------

The free dagger compact closed category on one generating object is
described concretely by Kelly and
Laplaza~\cite{AbramskyAS,coherence}.  The generator produces the
atomic positive interface $(1,0)$; duality gives $(0,1)$; tensor
builds all signed interfaces
$(m,n)=(1,0)^{\otimes m}\otimes(0,1)^{\otimes n}$.  Thus an object
of the Kelly--Laplaza base $\mathbf{KL}$ is a signed finite set
$M=(M^+,M^-)$, and a morphism $f:M\to N$ is a polarized bijection
$f^\sharp:M^+\sqcup N^-\cong N^+\sqcup M^-$ equipped with a loop
count $\kappa\in\NN$.

Composition is path execution with cycle formation: gluing along the
shared boundary decomposes each connected component into either a
path, contributing to the composite wiring, or a closed cycle; if
$\ell(g,f)$ denotes the number of new cycles, then
\[
\kappa(g\circ f) \;=\; \kappa(f)+\kappa(g)+\ell(g,f).
\]
Associativity and the correctness of the loop accounting follow from
Kelly--Laplaza coherence~\cite{AbramskyAS,coherence}.  The dagger is
inverse wiring; duality is polarity reversal, $M^*:=(M^-,M^+)$; the
tensor unit is $\Unit:=(\varnothing,\varnothing)$, the empty
interface, arising as the nullary tensor.

%-----------------------------------------------------------------------------
\subsection{The finite-coproduct precursor
\texorpdfstring{$\mathrm{Cop}$}{Cop}}
\label{subsec:ambient-coproducts}
%-----------------------------------------------------------------------------

We lift from individual interfaces to finite indexed families,
obtaining genuine finite coproducts.

\begin{definition}[Objects in $\oplus/\otimes$ normal form]
\label{def:distributed-objects}
An object in $\oplus/\otimes$ \emph{normal form} is a finite indexed
family
\[
A \;=\; \bigoplus_{i\in I} M_i,
\qquad M_i\in\mathbf{KL},
\]
with $I$ possibly empty (the empty family is the initial
object~$0$).  In family normal form the tensor unit is the singleton
family containing the empty KL interface~$\Unit$; it is distinct from
the empty-family zero object~$0$.  Tensor and objectwise polarity
reversal of normal-form objects are
\[
\Bigl(\bigoplus_{i\in I}M_i\Bigr)\otimes
\Bigl(\bigoplus_{j\in J}N_j\Bigr)
:=
\bigoplus_{(i,j)\in I\times J}(M_i\otimes N_j),
\qquad
\Bigl(\bigoplus_{i\in I}M_i\Bigr)^{*}
:=\bigoplus_{i\in I}M_i^{*},
\]
where at the port level
$(M^+,M^-)\otimes(N^+,N^-)=(M^+\sqcup N^+,\,M^-\sqcup N^-)$.
\end{definition}

At the index level, $\oplus$ is the $+$ of $\mathbf{FinBij}$
(concatenation of families) and $\otimes$ is its $\times$ (bilinear
distribution). For $A=\bigoplus_{i\in I}M_i$ and
$\epsilon\in\{+,-\}$, put
\[
 \operatorname{Ports}^{\epsilon}(A)
 :=\bigsqcup_{i\in I}M_i^{\epsilon},
 \qquad
 \Pos(A):=\operatorname{Ports}^{+}(A),
 \qquad
 \Neg(A):=\operatorname{Ports}^{-}(A).
\]
For a single interface,
$\operatorname{Ports}^{\epsilon}(M)=M^\epsilon$, extending the earlier
$\Pos/\Neg$ convention.

Every object of $\mathrm{Cop}$ and of its full $\Bbbk$-linear
matrix/biproduct completion is
in normal form, so $\otimes$ distributes over $\oplus$ by construction.
This indexing discipline prevents distributivity from being read as
port copying; see Remark~\ref{rem:normal-form-load-bearing}.
\begin{definition}[The finite-coproduct precursor \(\mathrm{Cop}\)]
\label{def:cop-pm}
$\mathrm{Cop}:=\mathrm{Cop}(\mathbf{KL})$ has as objects the
normal-form objects of Definition~\ref{def:distributed-objects}, and
as morphisms $f:A\to B$, for $A=\bigoplus_{i\in I}M_i$ and
$B=\bigoplus_{j\in J}N_j$, pairs
$f=\bigl(\varphi,(f_i)_{i\in I}\bigr)$ with
\begin{enumerate}[label=\textup{(\roman*)},nosep]
\item $\varphi:I\to J$ a function;
\item $f_i:M_i\to N_{\varphi(i)}$ a Kelly--Laplaza morphism, for
each $i\in I$.
\end{enumerate}
Composition is componentwise:
$\bigl(\psi,(g_j)\bigr)\circ\bigl(\varphi,(f_i)\bigr)
 :=\bigl(\psi\varphi,\,(g_{\varphi(i)}\circ f_i)\bigr)$.
\end{definition}

Equivalently, $f$ may be displayed as a $J\times I$ incidence array
with exactly one occupied KL entry in each column; the remaining
positions are empty --- not zero morphisms, which $\mathrm{Cop}$
does not have.  The canonical injections and cotuplings exhibit
$\mathrm{Cop}$ as the free finite-coproduct completion
of~$\mathbf{KL}$.

%-----------------------------------------------------------------------------
\subsection{The full \texorpdfstring{$\Bbbk$}{k}-linear matrix/biproduct completion \texorpdfstring{$\Perm(\Bbbk)$}{Perm(k)}}
\label{subsec:k-linear}
%-----------------------------------------------------------------------------

Equip $\Bbbk=\CC[\delta]$ with the involution that conjugates complex
coefficients and fixes $\delta$. Definition~\ref{def:perm-complex}
below specializes $\delta\mapsto d$ along $\CC[\delta]\to\CC$, for a
chosen $d\in\RR_{>0}\setminus\{1\}$.\footnote{The choice $d\neq1$ makes
closed trace loops detectable on a nonempty residual boundary.
Lemma~\ref{lem:final-multiplicative-base} shows that every correct
unit-free no-Mix multiplicative component used by the source calculus
is a coefficient-one matching with no detached execution component,
and hence contributes no power of $d$.}

\begin{definition}[Full $\Bbbk$-linear matrix/biproduct completion]
\label{def:perm-C}
Write $\mathbf{KL}^{0}(M,N)$ for the loop-free KL morphisms and $[f]$
for the basis vector of $f\in\mathbf{KL}^{0}(M,N)$. For
$A=\bigoplus_{i\in I}M_i$ and $B=\bigoplus_{j\in J}N_j$, set
\[
 \Perm(\Bbbk)(A,B)
 :=\mathrm{Mat}_{J\times I}
 \bigl(\Bbbk[\mathbf{KL}^{0}(M_i,N_j)]\bigr).
\]
Each entry is therefore a finite $\Bbbk$-linear combination of
loop-free matchings. For $f\in\mathbf{KL}^{0}(M,N)$ and
$g\in\mathbf{KL}^{0}(N,P)$, let $g\star f$ be the external-path
matching of their glued graph and let $\ell(g,f)$ be its number of
closed cycles. Define
\begin{equation}
 \label{eq:delta-twisted-basis-composition}
 [g]\circ[f]:=\delta^{\ell(g,f)}[g\star f],
\end{equation}
Extend this law $\Bbbk$-bilinearly to
$\mathcal L:=\Bbbk[\mathbf{KL}^{0}]$. Composition in $\Perm(\Bbbk)$
is matrix multiplication over $\mathcal L$. Dagger is conjugate transpose with
$[f]^\dagger=[f^\dagger]$ and $\bar\delta=\delta$. The linearization
of a general KL morphism is
\[
 \mathbf{KL}\longrightarrow\mathcal L,
 \qquad
 (f^\sharp,\kappa)\longmapsto\delta^{\kappa}[f^\sharp].
\]
\end{definition}

An occupied Kelly--Laplaza component is \emph{coefficient-one} if it is
a single basis matching with scalar \(1\). On pure-axis summands this
matching is a Kelly--Laplaza basis isomorphism. A matrix is
coefficient-one if every occupied entry is coefficient-one.

Brackets are suppressed when no confusion can result.  Explicitly,
if
\[
 F_{ji}=\sum_r\alpha_{ji,r}[f_{ji,r}],
 \qquad
 G_{kj}=\sum_s\beta_{kj,s}[g_{kj,s}],
\]
then matrix composition reads
\begin{equation}
 \label{eq:delta-twisted-matrix-composition}
 (G\circ F)_{ki}
 =\sum_{j,r,s}\beta_{kj,s}\alpha_{ji,r}
   \delta^{\ell(g_{kj,s},f_{ji,r})}
   [g_{kj,s}\star f_{ji,r}].
\end{equation}
Thus the image of $(\tau,\kappa)$ is $\delta^{\kappa}[\tau]$: the
loop count contributes to the coefficient, not to the basis label.

$\mathrm{Cop}\hookrightarrow\Perm(\Bbbk)$ on the same objects: its
image is the subcategory with exactly one nonzero entry per column,
that entry a single $\mathbf{KL}$ monomial $\delta^{\kappa}[\tau]$
supplied by the data of Definition~\ref{def:cop-pm}.  So
$\mathrm{Cop}$ pins the $\oplus$-indexing, and $\Perm(\Bbbk)$ adds
$\Bbbk$-linear combinations in each entry.  Equivalently, one may
first form finite-family objects and then take full matrix
hom-spaces over $\mathcal L$, recovering the same hom-space: the
coproduct-first view fixes the indexing, the $\mathcal L$-first view
exposes the basis rule~\eqref{eq:delta-twisted-basis-composition}.
Once compact closure has been verified, its canonical trace recovers
the same gluing rule.

\extendedonly{%
For a pair of loop-free KL basis morphisms, the gluing can be computed
by a finite transient-path sum after its closed internal cycles have
been extracted.  No inverse over \(\Bbbk\) is used;
Appendix~\ref{app:structural-proofs} gives the basiswise path-execution
proof.%
}
\begin{proposition}[The loop-erased envelope]
\label{prop:well-defined-assoc}
The rule~\eqref{eq:delta-twisted-basis-composition} makes the
loop-erased $\Bbbk$-linear envelope
$\mathcal L=\Bbbk[\mathbf{KL}^{0}]$ a dagger compact closed category.
Its compact structure supplies the canonical Joyal--Street--Verity
trace, the linearization $\mathbf{KL}\to\mathcal L$ is functorial,
and the matrix composition~\eqref{eq:delta-twisted-matrix-composition}
makes $\Perm(\Bbbk)$ a dagger compact closed category with finite
coproducts.
\end{proposition}

\begin{proof}
Appendix~\ref{app:structural-proofs} gives the direct componentwise
proof of the category laws, compact equations, functoriality of the
linearization, and the basis-level feedback calculation.
\end{proof}

Its tensor is inherited from $\mathbf{KL}$, while its
finite-coproduct indexing is governed by $\mathbf{FinBij}$; every
object remains in the $\oplus/\otimes$ normal form of
Definition~\ref{def:distributed-objects}.

\paragraph*{Biproducts in $\Perm(\Bbbk)$.}
$\Perm(\Bbbk)$ is $\Bbbk$-linear with zero object (the empty
family), so by the standard semiadditivity argument finite
coproducts are biproducts: $A\oplus B$ carries injections $\iota_i$
and projections $\pi_i$ satisfying $\pi_i\circ\iota_j=\id$ for
$i=j$, $\pi_i\circ\iota_j=0$ for $i\neq j$, and
$\sum_i\iota_i\circ\pi_i=\id_{A\oplus B}$.  Biproducts are not
available in $\mathrm{Cop}$, where the one-occupied-entry-per-column
condition rules out product projections.  The source calculus
introduced later does not expose biproduct injections or projections
as source constructors; such maps remain available for
carrier-level matrix arguments in $\Perm(\CC)$.

\begin{remark}[Currying]
\label{rem:curry-invisible}
A morphism $f:A\otimes B\to C$ and its curry
$\tilde f:A\to B^*\otimes C$ share the same underlying wiring:
currying merely repartitions boundary ports.
On basis morphisms this is the identity on polarized bijections;
it extends $\Bbbk$-linearly to all of $\Perm(\Bbbk)$.
\end{remark}

\begin{definition}[$\Perm(\CC)$]
\label{def:perm-complex}
Fix $d\in\RR_{>0}\setminus\{1\}$ and let
$\operatorname{ev}_d:\CC[\delta]\to\CC$ be evaluation at
$\delta=d$.  Then $\Perm(\CC)$ is the coefficientwise specialization
of $\Perm(\Bbbk)$ along $\operatorname{ev}_d$: it has the same
objects, and its hom-spaces are obtained by applying
$\operatorname{ev}_d$ to every coefficient.
\end{definition}

Because $d$ is real, the evaluation preserves the involution;
consequently $\Perm(\CC)$ is a $\CC$-linear dagger compact closed
category with tensor unit $\Unit$ and finite biproducts $\oplus$.
By abuse of notation we continue to write $\delta$ for the chosen
scalar~$d$.  For a finite set $S$, write $\CC[S]$ for the
finite-dimensional Hilbert space with orthonormal basis $S$; this
set-linearization notation is distinct from the polynomial ring
$\CC[\delta]$.

Every morphism $F:A\to B$ of $\Perm(\CC)$, with
$A=\bigoplus_{i\in I}M_i$ and $B=\bigoplus_{j\in J}N_j$, is by
definition a $J\times I$ matrix
\[
  F=(F_{ji}),\qquad F_{ji}\in\CC[\mathbf{KL}^{0}(M_i,N_j)].
\]
Writing each entry uniquely as a finite collected sum of distinct
loop-free KL matchings,
\[
  F_{ji}=\sum_r \alpha_{ji,r}\,g_{ji,r},
  \qquad g_{ji,r}\in\mathbf{KL}^{0}(M_i,N_j),
\]
gives the \emph{collected KL matrix form} of $F$; each coefficient
$\alpha_{ji,r}$ already contains the evaluated powers of $d$
produced by closed cycles, so loop counts are not additional basis
data.

\begin{definition}[First-order objects]
\label{def:first-order}
An object $A=\bigoplus_{i\in I}M_i$ is \emph{first-order} if all its
summands are pure-sign of the same polarity: $M_i=(m_i,0)$ for all
$i$, or $M_i=(0,m_i)$ for all $i$.  A morphism is
\emph{first-order} if its domain and codomain are first-order of the
same polarity.  Write $\mathbf{H}_n := (1,0)^{\oplus n}$ for the
first-order object of \emph{size}~$n$.
\end{definition}

Fix either polarity.  In the positive sector there are no negative
ports (the negative sector is its dual/opposite), so no new closed cycles
form under composition: the $\delta$-twist disappears, and
composition is the usual matrix-sum formula with entry
multiplication given by composition of loop-free KL matchings.  For
the one-port objects $\mathbf H_n$ the matching spaces are
one-dimensional, and this reduces to ordinary complex matrix
multiplication.

Section~\ref{sec:boundary-transport} first assigns boundary matrices
to morphisms between individual KL interfaces, then incorporates
finite-family qudit addresses in fixed-width qudit--MLL matrices and
uniform sequent matrices. The $\oplus/\otimes$ normal form and KL
matrix-unit basis make each stage unambiguous.

\section{Kelly--Laplaza Boundaries and Essential Unitarity}
\label{sec:boundary-transport}

Boundary transport reads a Kelly--Laplaza morphism \(M\to N\) on the
two polarities of its compact-name interface \(M^*\otimes N\). A basis
morphism contributes the permutation matrix of its axiom-link matching;
linear combinations contribute the corresponding \(\CC\)-linear
combination, with evaluated loop factors already present in the
coefficients. The type repartitions the boundary; the wiring is
unchanged. On the one-port registers \(\mathbf H_n\), this
recovers ordinary complex matrices. Qudits tensored with multiplicative
payloads will give the full matrices used later.

%-----------------------------------------------------------------------------
\subsection{KL boundary matrices}
\label{subsec:boundary-ports}
%-----------------------------------------------------------------------------

\begin{definition}[KL boundary matching]
\label{def:boundary-objects}
\label{def:boundary-bijection}
For signed KL interfaces \(M,N\), put
\[
 P_{M,N}:=\Neg(M)\sqcup\Pos(N),
 \qquad
 N_{M,N}:=\Pos(M)\sqcup\Neg(N).
\]
If \(g:M\to N\) has polarized wiring
\[
 g^\sharp:N_{M,N}\xrightarrow{\cong}
 \Pos(N)\sqcup\Neg(M),
\]
let \(\chi_{M,N}\) be the canonical swap onto \(P_{M,N}\), and set
\[
 \tau_g:=\chi_{M,N}g^\sharp:
 N_{M,N}\xrightarrow{\cong}P_{M,N}.
\]
Write
\(\underline{\tau_g}:\CC[N_{M,N}]\to\CC[P_{M,N}]\)
for its permutation matrix.
\end{definition}

We read the boundary covariantly, from \(\CC[N_{M,N}]\) to
\(\CC[P_{M,N}]\): \(g^\sharp\) is used without inversion, and its
codomain is placed in the displayed \(P\)-order.

\begin{remark}[Axiom links and the compact-name boundary]
\label{rem:axiom-link-extraction}
The map \(\tau_g\) is the same axiom-link wiring reorganized around
\(M^*\otimes N\). It is therefore square even when \(M\) and \(N\)
have different shapes. In particular, evaluation has differently
shaped source and target but a boundary consisting of straight wires.
Currying likewise preserves each component matching under the canonical
identification of compact-name ports
(Remark~\ref{rem:curry-invisible}).
\end{remark}

Every \(F:M\to N\) in \(\Perm(\CC)\) has a unique KL-basis
expansion \(F=\sum_g c_g[g]\). Its coefficients already include all
evaluated loop factors.

\begin{definition}[KL boundary matrix]
\label{def:KL-boundary-matrix}
Define
\[
 \mathsf B_{M,N}(F)
 :=\sum_g c_g\,\underline{\tau_g}:
 \CC[N_{M,N}]\longrightarrow\CC[P_{M,N}].
\]
\end{definition}

\begin{definition}[Essential unitarity]
\label{def:EU}
A morphism \(F:M\to N\) is \emph{essentially unitary} if
\[
 \mathsf B_{M,N}(F)^\dagger\mathsf B_{M,N}(F)
 =I_{\CC[N_{M,N}]},
 \qquad
 \mathsf B_{M,N}(F)\mathsf B_{M,N}(F)^\dagger
 =I_{\CC[P_{M,N}]}.
\]
Write \(\operatorname{EU}(F)\) for this condition.
\end{definition}

Every loop-free KL basis morphism is EU. A component with loop-free
part \(g\) and loop count \(\kappa\) has boundary matrix
\(d^\kappa\underline{\tau_g}\); on a nonempty boundary it is EU exactly
when \(\kappa=0\). The empty boundary cannot detect scalars.

This boundary condition is distinct from carrier dagger-unitarity. For
nonempty \(M\), the evaluation component
\[
 \ev_{M,N}:(M^*\otimes N)\otimes M\longrightarrow N
\]
is EU because its boundary matrix is a permutation, but
\[
 \ev_{M,N}\ev_{M,N}^\dagger
 =d^{|M^+|+|M^-|}\id_N\ne\id_N
\]
in \(\Perm(\CC)\).

The boundary representation also need not be injective on linear
combinations. For \(R_3:=(3,0)\), put
\[
 K:=\sum_{\sigma\in S_3}\operatorname{sgn}(\sigma)[\sigma].
\]
Then \(K\ne0\), whereas
\(\mathsf B_{R_3,R_3}(K)=0\): at each matrix entry, one even and one
odd permutation cancel. Hence \(\id_{R_3}+\lambda K\) is EU for every
\(\lambda\in\CC\). Since \(K^\dagger=K\) and \(K^2=6K\),
\[
 (\id_{R_3}+\lambda K)^\dagger(\id_{R_3}+\lambda K)
 =\id_{R_3}
  +(\lambda+\bar\lambda+6|\lambda|^2)K,
\]
so it is generally not carrier dagger-unitary. Thus EU records the
induced action on the polarized port space, not the entire carrier
element.

%-----------------------------------------------------------------------------
\subsection{Execution and loop coefficients}
\label{subsec:composition-law}
%-----------------------------------------------------------------------------

Because the boundary retains external paths and absorbs cycles into
coefficients, composition is governed by KL execution.

\begin{proposition}[KL execution]
\label{prop:boundary-execution-gluing}
For morphisms \(f:M\to N\) and \(g:N\to P\) in \(\mathbf{KL}\),
the external paths
determine \(\tau_{g\circ f}\), while
\[
 \kappa(g\circ f)
 =\kappa(f)+\kappa(g)+\ell(g,f),
\]
where \(\ell(g,f)\) is the number of new internal cycles.
\end{proposition}

This is the paths-and-cycles decomposition of
\cite[Proposition~1]{AbramskyAS}; the port calculation is in
Appendix~\ref{app:boundary-proofs}. For mixed interfaces it is a
geometric execution law, not multiplication of the square boundary
matrices: the output partition of \(f\) need not be the input partition
of \(g\).

\begin{lemma}[Dagger]
\label{lem:KL-boundary-dagger}
For every \(F:M\to N\),
\[
 \operatorname{EU}(F^\dagger)\iff\operatorname{EU}(F).
\]
\end{lemma}

\begin{proof}
On basis morphisms, dagger inverts the matching and exchanges the two
boundary orders. These exchanges are unitary permutation matrices;
conjugating coefficients gives the two Gram identities.
\end{proof}

\begin{remark}[Cup--cap loops break closure]
\label{rem:eta-epsilon-EU}
For nonempty \(M\), the unit \(\eta_M\) and the composable counit
\(\epsilon_{M^*}\) are loop-free and hence EU, whereas
\[
 s_M:=\epsilon_{M^*}\eta_M
 =d^{|M^+|+|M^-|}\id_\Unit.
\]
For every nonempty \(X\),
\[
 (\epsilon_{M^*}\otimes\id_X)
 (\eta_M\otimes\id_X)
 =s_M\otimes\id_X
\]
has boundary matrix
\(d^{|M^+|+|M^-|}\mathsf B_{X,X}(\id_X)\).
Thus EU is not closed under composition or tensor in \(\Perm(\CC)\).
\end{remark}

%=============================================================================
\section{Qudits, MLL Payloads, and Full Matrices}
\label{sec:eu-predicate}
%=============================================================================

On a pure-sign intermediate, KL execution becomes ordinary matrix
composition. Naming and polarized un-naming place a qudit with an
arbitrary multiplicative payload on such pure-positive axes; this
assignment is not functorial for arbitrary composition through a mixed
payload. These two steps give the full matrix used by the source
calculus.

%-----------------------------------------------------------------------------
\subsection{Positive qudit matrices}
%-----------------------------------------------------------------------------

For \(m,n\ge1\), write \([n]=\{1,\ldots,n\}\), and similarly
for \([m]\). Recall that
\(\mathbf H_n=(1,0)^{\oplus n}\) and that \(\CC[S]\) has orthonormal
basis \(S\). A morphism
\(U:\mathbf H_n\to\mathbf H_m\) is therefore its ordinary coefficient
matrix
\[
 U=\sum_{j,i}U_{ji}e_{ji}:\CC[n]\longrightarrow\CC[m],
 \qquad
 e_{ji}:=\lvert j\rangle\!\langle i\rvert.
\]
Carrier dagger-unitarity and ordinary matrix unitarity coincide here.

The same formula retains a nontrivial KL matching at each matrix entry.
Let \(R_r:=(r,0)\) for \(r\ge0\). Fix \(r\ge1\), and write
\[
 G_{ji}=\sum_{g\in\mathbf{KL}^{0}(R_r,R_r)}c_{ji,g}[g]
\]
for \(G:R_r^{\oplus n}\to R_r^{\oplus m}\).

\begin{definition}[Fixed-width qudit--KL matrix]
\label{def:first-order-boundary-matrix}
Define
\[
 \mathsf B_r(G)
 :=\sum_{j,i,g}c_{ji,g}\,
 e_{ji}\otimes\underline{\tau_g}:
 \CC[n]\otimes\CC[r]\longrightarrow
 \CC[m]\otimes\CC[r].
\]
\end{definition}

The two coordinates have different roles: \(i,j\) are qudit indices,
while \(\underline{\tau_g}\) acts on the \(r\) KL ports.

\begin{remark}[Distributivity and family indices]
\label{rem:normal-form-load-bearing}
Within one pure-sign sector, normal form keeps biproduct-family indices
separate from KL ports. Carrier distributivity only reindexes the
families; its KL components are identity matchings, so its matrix is a
unitary permutation. No port is copied.
\end{remark}

\begin{proposition}[Fixed-width representation]
\label{prop:first-order-representation}
For composable \(G,H\) in the fixed \(r\)-sector,
\[
 \mathsf B_r(\id)=I,\qquad
 \mathsf B_r(HG)=\mathsf B_r(H)\mathsf B_r(G),\qquad
 \mathsf B_r(G^\dagger)=\mathsf B_r(G)^\dagger.
\]
\end{proposition}

\begin{proof}
A pure-positive intermediate has no opposite-polarity ports, so
execution creates no cycle. KL composition is ordinary permutation
composition, and the result is block-matrix multiplication.
\end{proof}

For \(r=1\), the matching space is one-dimensional and
\(\mathsf B_1(U)=U\).

%-----------------------------------------------------------------------------
\subsection{Qudits with MLL payloads}
%-----------------------------------------------------------------------------

For \(f:U\to V\) and \(a:\Unit\to U^*\otimes V\), use
\[
 \operatorname{name}(f):=(\id_{U^*}\otimes f)\eta_U,
 \qquad
 \operatorname{unname}(a):=
 (\epsilon_U\otimes\id_V)(\id_U\otimes a),
\]
where \(\eta_U:\Unit\to U^*\otimes U\) and
\(\epsilon_U:U\otimes U^*\to\Unit\). Fixed associativity and unit
coherence are suppressed; name and unname are inverse.

Let \(M=(M^+,M^-)\) and \(N=(N^+,N^-)\) be carrier KL
interfaces. For \(Y\in\{M,N\}\), put
\[
 Y_+:=(Y^+,\varnothing),
 \qquad
 Y_-:=(Y^-,\varnothing).
\]
These pure-positive interfaces give the canonical coefficient-one
factorization
\[
 Y\cong Y_-^*\otimes Y_+.
\]
Consequently, for
\[
 A:=\mathbf H_n\otimes M,
 \qquad
 B:=\mathbf H_m\otimes N,
\]
there are canonical displays
\[
 A\cong M_-^*\otimes(\mathbf H_n\otimes M_+),
 \qquad
 B\cong N_-^*\otimes(\mathbf H_m\otimes N_+).
\]

These displays explain the full qudit--MLL matrix. They move the qudit
index to the positive axis before the compact name is uncurried.

Let
\[
 \zeta_{M,N}^{n,m}:
 A^*\otimes B\xrightarrow{\cong}
 \bigl(\mathbf H_n\otimes(M_+\otimes N_-)\bigr)^*
 \otimes
 \bigl(\mathbf H_m\otimes(M_-\otimes N_+)\bigr)
\]
be the canonical coefficient-one compact-structural isomorphism
induced by duality, symmetry, and associativity.

\begin{definition}[Qudit--MLL carrier matrix]
\label{def:qudit-mll-carrier-matrix}
For \(F:A\to B\), define
\[
 \mathcal Q_{M,N}^{n,m}(F)
 :=
 \operatorname{unname}\!\left(
 \zeta_{M,N}^{n,m}
 \operatorname{name}_{\Perm(\CC)}(F)\right).
\]
This is a positive first-order carrier morphism from
\(\mathbf H_n\otimes(M_+\otimes N_-)\) to
\(\mathbf H_m\otimes(M_-\otimes N_+)\).
\end{definition}

This type supports a KL basis component precisely when
\[
 1+|M^+|+|N^-|=1+|M^-|+|N^+|.
\]
Assume this balance and write the common value as \(r\).
The extra \(1\) is the qudit's own KL port. After distribution,
\(\mathcal Q_{M,N}^{n,m}(F)\) lies in the fixed-width sector
\(R_r^{\oplus n}\to R_r^{\oplus m}\). Its full numerical matrix is
\[
 \mathbf Q_{M,N}^{n,m}(F)
 :=\mathsf B_r\!\left(\mathcal Q_{M,N}^{n,m}(F)\right):
 \CC[n]\otimes\CC[r]\longrightarrow\CC[m]\otimes\CC[r].
\]
Equivalently, if
\[
 \mathcal Q_{M,N}^{n,m}(F)_{ji}
 =\sum_g c_{ji,g}[g],
\]
then
\[
 \mathbf Q_{M,N}^{n,m}(F)
 =\sum_{j,i,g}c_{ji,g}\,
 e_{ji}\otimes\underline{\tau_g}.
\]

\begin{corollary}[Qudit--MLL unitarity]
\label{cor:qudit-mll-unitarity}
If \(\mathcal Q_{M,N}^{n,m}(F)\) is carrier dagger-unitary, then
\(\mathbf Q_{M,N}^{n,m}(F)\) is an ordinary two-sided unitary.
\end{corollary}

In particular, let \(U\in\mathsf U(n)\) and let \(p:M\to N\) be a
coefficient-one loop-free KL basis morphism. Up to the fixed boundary
permutations,
\[
 \mathbf Q_{M,N}^{n,n}(U\otimes p)
 =
 U\otimes\underline{\tau_{\id_{R_1}\otimes p}}
 =
 U\otimes\underline{\bigl(\id_{\{q\}}\sqcup\tau_p\bigr)},
\]
where \(\{q\}\) is the singleton qudit-port set. Thus the qudit and MLL
factors are jointly unitary. For general linear
combinations in this fixed block format, unitarity is a condition on
the full block
matrix; unitary constituent matchings alone do not suffice.

%-----------------------------------------------------------------------------
\subsection{Uniform qudit--MLL boundaries}
%-----------------------------------------------------------------------------

The preceding construction iterates over any nonempty one-sided
qudit--MLL boundary.  For a carrier formula \(X\) generated from
\(\mathbf H_n\), \(\mathbf H_n^*\), \(\otimes\), and \(\parr\), collect
its qudit rows by polarity:
\[
 \begin{aligned}
 \mathsf A_{\mathbf H_n}^-&:=\Unit,
 &\mathsf A_{\mathbf H_n}^+&:=\mathbf H_n,\\
 \mathsf A_{\mathbf H_n^*}^-&:=\mathbf H_n,
 &\mathsf A_{\mathbf H_n^*}^+&:=\Unit,\\
 \mathsf A_{X\diamond Y}^{\epsilon}
 &:=\mathsf A_X^{\epsilon}\otimes\mathsf A_Y^{\epsilon}
 &&(\diamond\in\{\otimes,\parr\},\ \epsilon\in\{-,+\}).
 \end{aligned}
\]
For a sequent \(\Sigma=(X_1,\ldots,X_s)\), tensor the corresponding
factors and put \(\mathsf A_\varnothing^\pm:=\Unit\).  Write
\(|\Sigma|:=X_1\otimes\cdots\otimes X_s\) for its carrier object
and \(|\varnothing|:=\Unit\).  Each polarized factor is pure positive.
Its uniform-row display is a fixed coefficient-one dagger
isomorphism:
\[
 d_\Sigma^\epsilon:
 \mathsf A_\Sigma^\epsilon
 \xrightarrow{\cong}
 R_{k_\Sigma^\epsilon}^{\oplus N_\Sigma^\epsilon}.
\]
Here \(k_\Sigma^\epsilon\) is its port width and
\(N_\Sigma^\epsilon\) its number of qudit-address rows; for the empty
factor these are \(0\) and \(1\).

The fixed compact coherence gives a coefficient-one dagger isomorphism
\[
 \zeta_\Sigma:
 |\Sigma|\xrightarrow{\cong}
 (\mathsf A_\Sigma^-)^*\otimes\mathsf A_\Sigma^+.
\]
This is the preceding qudit--MLL polarization applied simultaneously
to all literals of the boundary.

\begin{definition}[Qudit--MLL sequent matrix]
\label{def:qudit-mll-sequent-matrix}
For \(a:\Unit\to|\Sigma|\), define
\[
 \mathcal Q_\Sigma(a)
 :=\operatorname{unname}(\zeta_\Sigma a):
 \mathsf A_\Sigma^-\longrightarrow\mathsf A_\Sigma^+.
\]
If \(k_\Sigma^-=k_\Sigma^+=r\geq1\), its full numerical realization is
\[
 \mathbf Q_\Sigma(a)
 :=\mathsf B_r\!\left(
 d_\Sigma^+\mathcal Q_\Sigma(a)(d_\Sigma^-)^{-1}
 \right).
\]
\end{definition}

The sequent construction is the simultaneous iteration of
Definition~\ref{def:qudit-mll-carrier-matrix} over all boundary
literals.

\begin{definition}[Qudit--MLL essential unitarity]
\label{def:qudit-mll-EU}
For a morphism
\(F:\mathbf H_n\otimes M\to\mathbf H_m\otimes N\) satisfying
\[
 1+|M^+|+|N^-|=1+|M^-|+|N^+|,
\]
qudit--MLL essential unitarity means that
\(\mathbf Q_{M,N}^{n,m}(F)\) is a two-sided unitary. For a sequent
point \(a:\Unit\to|\Sigma|\) with
\(k_\Sigma^-=k_\Sigma^+\geq1\), it means that
\(\mathbf Q_\Sigma(a)\) is a two-sided unitary.
\end{definition}

Carrier dagger-unitarity of \(\mathcal Q_\Sigma(a)\) implies the
sequent condition by Proposition~\ref{prop:first-order-representation}.

%-----------------------------------------------------------------------------
\subsection{Limits of the qudit--MLL matrix}
%-----------------------------------------------------------------------------

The qudit--MLL construction is not functorial for arbitrary composition
through a mixed payload, and such composition need not preserve
qudit--MLL essential unitarity.
Proposition~\ref{prop:omg-prime-h-witness} gives a concrete six-port
witness whose three factors have unitary matrices but whose composite
does not; Appendix~\ref{app:QC-h-witness} gives the execution
calculation.

A distinct example distinguishes boundary EU from carrier
dagger-unitarity on a single mixed interface. Let \(X=(1,0)\),
\(A=X\otimes X^*\), and label the source and target ports by
\(x_{\mathrm{in}}^\pm\) and \(x_{\mathrm{out}}^\pm\). Let
\(E:A\to A\) be the cap--cup turnback with links
\(x_{\mathrm{in}}^+\text{--}x_{\mathrm{in}}^-\) and
\(x_{\mathrm{out}}^-\text{--}x_{\mathrm{out}}^+\). Then
\[
 E^\dagger=E,\qquad E^2=dE,\qquad
 F:=\id_A-\frac{2}{d}E
\]
is carrier dagger-unitary. In the natural two-port boundary orders,
\[
 \mathsf B_{A,A}(\id_A)
 =\begin{pmatrix}0&1\\1&0\end{pmatrix},
 \qquad
 \mathsf B_{A,A}(E)=I_2,
\]
whereas
\[
 \mathsf B_{A,A}(F)^\dagger\mathsf B_{A,A}(F)
 =
 \begin{pmatrix}
  1+4/d^2&-4/d\\
  -4/d&1+4/d^2
 \end{pmatrix}
 \ne I_2.
\]
Thus carrier dagger-unitarity implies full matrix unitarity on the
pure axes used later, but not on arbitrary mixed objects of
\(\Perm(\CC)\).

\section{The Coherent Syntax}
\label{sec:QC}

This section defines the source category \(\QC\), its carrier
interpretation, its unitary gates on positive qudit registers, and its
contextual control constructor. The qudit--MLL matrix interface is
developed in Section~\ref{sec:eu-predicate}; essential unitarity is
proved in Section~\ref{sec:calculus}.

%-----------------------------------------------------------------------------
\subsection{Why Qudit Control}
\label{subsec:QC-why-coherence}
%-----------------------------------------------------------------------------

Two carrier counterexamples in \(\Perm(\CC)\) motivate the source
restrictions below. For the first, put
\[
 Q:=(1,0)=\mathbf H_1,\qquad
 F:=Q^*\otimes Q,\qquad
 M:=F\otimes Q.
\]
In each copy of \(M\), label the positive function-output and payload
ports by \(a,b\), and the negative function-input port by \(c\);
primes mark the target copy. Let \(\bar w:M\to M\) be the
coefficient-one multiplicative basis morphism with links
\[
 a\longmapsto b',
 \qquad
 b\longmapsto c,
 \qquad
 c'\longmapsto a'.
\]
Let \(X:M\oplus M\to M\oplus M\) exchange the two summands, set
\(L:=\id_M\oplus\bar w\), and, for \(\theta\in\mathbb R\), put
\[
 h:=L\circ
 \bigl(\cos\theta\,\id_{M\oplus M}+\ii\sin\theta\,X\bigr)
 \circ L.
\]
Let
\[
 \alpha:\mathbf H_2\otimes F\xrightarrow{\cong}M\oplus M
\]
be the canonical coefficient-one carrier isomorphism, and define
\[
 \mathsf C(K):=
 \mathbf Q_{F,F}^{2,2}(\alpha^{-1}K\alpha)
 \qquad
 (K:M\oplus M\to M\oplus M).
\]

\begin{proposition}[The motivating witness \(h\)]
\label{prop:omg-prime-h-witness}
If \(\theta\notin\tfrac{\pi}{2}\mathbb Z\), then
\[
 \mathsf C(L),\
 \mathsf C\!\left(
   \cos\theta\,\id_{M\oplus M}+\ii\sin\theta\,X
 \right)
 \in\mathsf U(6),
 \qquad
 \mathsf C(h)\notin\mathsf U(6).
\]
\end{proposition}
\extendedonly{The execution calculation is given in
Appendix~\ref{app:QC-h-witness}.}

Thus arbitrary carrier-\(\oplus\) control is not stable over a mixed
payload, motivating explicit selector retention. The contextual
constructor below also requires a positive common output, ensuring the
pure-axis block form used by the compositionality and normalization
arguments.

For the second counterexample, put \(P:=(1,0)=\mathbf H_1\),
\(Y:=P\otimes P^*\), and \(Z:=Y\otimes Y\), and let
\(\sigma:Z\to Z\) exchange the two \(Y\)-factors.

\begin{proposition}[The mixed-interface witness]
\label{prop:QC-F2-witness}
There is a coefficient-one multiplicative basis morphism
\(\bar q:Z\to Y\) such that, for every
\(\theta\notin\tfrac{\pi}{2}\mathbb Z\),
\[
 \neg\operatorname{EU}\!\left(
   \bar q\circ
   \bigl(\cos\theta\,\id_Z+\ii\sin\theta\,\sigma\bigr)
 \right).
\]
\end{proposition}

\extendedonly{The proof is given in
Appendix~\ref{app:QC-h-witness}.}

Although the parenthesized carrier map is dagger-unitary, composition
can destroy its boundary unitarity because \(Z\) is mixed. Primitive
unitaries are therefore admitted on positive qudit registers, not on
arbitrary multiplicative interfaces.

%-----------------------------------------------------------------------------
\subsection{Raw Types and Qudit Registers}
\label{subsec:QC-objects}
%-----------------------------------------------------------------------------

\begin{definition}[Raw types]
\label{def:QC-basic-types}
For \(n\geq1\), let \(\mathbf R_n\) be the positive \(n\)-level qudit
literal and put \(q:=\mathbf R_1\). Raw types are generated by
\[
 A::=\mathbf R_n
 \mid \mathbf R_n^*
 \mid A\otimes B
 \mid A\parr B.
\]
Their strict De Morgan duality is
\[
 (\mathbf R_n^*)^*=\mathbf R_n,
 \qquad
 (A\otimes B)^*=A^*\parr B^*,
 \qquad
 (A\parr B)^*=A^*\otimes B^*.
\]
The polarity set is determined by
\[
 \begin{aligned}
 \operatorname{pol}(\mathbf R_n)&=\{+\},
 &
 \operatorname{pol}(\mathbf R_n^*)&=\{-\},\\
 \operatorname{pol}(A\otimes B)
 &=\operatorname{pol}(A\parr B)
 =\operatorname{pol}(A)\cup\operatorname{pol}(B).
 \end{aligned}
\]
A type is \emph{positive}, \emph{negative}, or \emph{mixed} according
as its polarity set is \(\{+\}\), \(\{-\}\), or \(\{+,-\}\).
A positive or negative type is \emph{pure}.
\end{definition}

The source grammar is unit-free: it has no multiplicative-unit type or
nullary compact constructor. The carrier category \(\Perm(\CC)\)
remains compact and unital. With its compact \(*\)-autonomous
structure, both source multiplicatives are interpreted by tensor.

\begin{definition}[Object assignment]
\label{def:QC-object-interpretation}
The raw-object assignment \(O\) is determined by
\[
 \begin{aligned}
 O(\mathbf R_n)&=\mathbf H_n=(1,0)^{\oplus n},
 & O(\mathbf R_n^*)&=\mathbf H_n^*,\\
 O(A\otimes B)&=O(A)\otimes O(B),
 & O(A\parr B)&=O(A)\otimes O(B).
 \end{aligned}
\]
\end{definition}

The source qudit is not a multiplicative unit: every summand of
\(\mathbf H_n\) carries one positive KL port. In particular,
\(\mathbf R_m\otimes\mathbf R_n\) and \(\mathbf R_{mn}\) have equally
many carrier rows but different multiplicative interfaces.

Write \([n]:=\{1,\ldots,n\}\). For positive types, define the address
set \(I_P\) and port width \(k_P\) recursively by
\[
 \begin{aligned}
 I_{\mathbf R_n}&:=[n],
 &
 k_{\mathbf R_n}&:=1,\\
 I_{P\diamond S}&:=I_P\times I_S,
 &
 k_{P\diamond S}&:=k_P+k_S
 \qquad(\diamond\in\{\otimes,\parr\}),
 \end{aligned}
\]
and put \(N_P:=|I_P|\). The carrier \(\oplus/\otimes\) normal form
supplies the fixed lexicographic qudit-row display
\begin{equation}
 d_P:O(P)\xrightarrow{\cong}
 \bigoplus_{\mathbf i\in I_P}(k_P,0).
 \label{eq:QC-positive-display}
\end{equation}
Thus \(I_P\) records qudit addresses, whereas each row retains all
\(k_P\) multiplicative ports.

\begin{definition}[Positive-register gates]
\label{def:QC-register-unitaries}
Positive types \(P,Q\) are \emph{shape-compatible} when
\(N_P=N_Q\) and \(k_P=k_Q\). For every shape-compatible pair and every unitary
\(U:\CC[I_P]\to\CC[I_Q]\), introduce
\(u_U^{P,Q}:P\to Q\), subject to
\[
 u_I^{P,P}=\id_P,
 \qquad
 u_V^{Q,R}\circ u_U^{P,Q}=u_{VU}^{P,R}.
\]
Negative register gates are their duals.
\end{definition}

Write \(u_U^P:=u_U^{P,P}\). For \(P=\mathbf R_n\), write
\(u_U:=u_U^{\mathbf R_n}\). Under the carrier assignment below,
\(u_U\) has ordinary coefficient matrix \(U\). Allowing \(P\) to be
compound supplies joint,
including entangling, unitaries on
\(\mathbf R_m\otimes\mathbf R_n\) without identifying that two-port
type with the one-port literal \(\mathbf R_{mn}\).

%-----------------------------------------------------------------------------
\subsection{The Category \texorpdfstring{\(\QC\)}{QC}}
\label{subsec:QC-category}
%-----------------------------------------------------------------------------

The remaining constructor realizes a synchronized family with one
qudit selector.  For a finite raw context
\(\Gamma=(A_1,\ldots,A_k)\), write
\[
 \Gamma\mathbin{\cdot}D
 :=A_1\otimes\cdots\otimes A_k\otimes D,
 \qquad
 \varnothing\mathbin{\cdot}D:=D,
\]
using the fixed association convention.

Let \(\mathsf{Obj}_{\mathrm u}\) be the formulas generated from the
qudit literals by \(\Unit,\bot,\otimes,\parr\), and strict De Morgan
duality, with \(\Unit^*=\bot\) and \(\bot^*=\Unit\).  Unit
occurrences remain visible in these formulas; the raw types are exactly
the unit-free ones.  We use the typed-component presentation
of~\cite{BCST1996}; the category-labelled construction
of~\cite{Hughes2012} is its ordinary-category special case.

\begin{definition}[Contextual qudit control]
\label{def:basic-admissible-coherent-tuple}
For \(X,Y\in\mathsf{Obj}_{\mathrm u}\), let
\(\mathsf{Term}(X,Y)\) be the least family of finite typed terms
generated by the unital symmetric \(*\)-autonomous term formers,
the positive-register gates, and
\begin{equation}
 \frac{f_i:\Gamma\mathbin{\cdot}D\to P\quad(i\in[n])}
      {\operatorname{ctrl}^{\Gamma}_n(f_1,\ldots,f_n):
       \mathbf R_n\otimes(\Gamma\mathbin{\cdot}D)
       \to\mathbf R_n\otimes P},
 \qquad n\geq1,\quad D,P\ \text{positive}.
 \label{eq:QC-control-formation}
\end{equation}
Put \(\operatorname{ctrl}_n:=\operatorname{ctrl}^{\varnothing}_n\).

Let \(\equiv\) be the least typed congruence, closed under every term
former, that contains the symmetric \(*\)-autonomous equations, the
register equations, and the following relations.  For
\[
 \begin{gathered}
 f,f_i:\Gamma\mathbin{\cdot}D\to P,
 \qquad g_i:P\to R,\\
 h:\Gamma'\mathbin{\cdot}D'\to\Gamma\mathbin{\cdot}D,
 \qquad D,D',P,R\ \text{positive},
 \end{gathered}
\]
impose
\begin{align}
 \operatorname{ctrl}^{\Gamma}_n(f,\ldots,f)
 &\equiv\id_{\mathbf R_n}\otimes f,
 \label{eq:QC-control-uniform}\\
 \operatorname{ctrl}_n(g_1,\ldots,g_n)
 \circ\operatorname{ctrl}^{\Gamma}_n(f_1,\ldots,f_n)
 &\equiv\operatorname{ctrl}^{\Gamma}_n
 (g_1\circ f_1,\ldots,g_n\circ f_n),
 \label{eq:QC-control-composition}\\
 \operatorname{ctrl}^{\Gamma}_n(f_1,\ldots,f_n)
 \circ(\id_{\mathbf R_n}\otimes h)
 &\equiv\operatorname{ctrl}^{\Gamma'}_n
 (f_1\circ h,\ldots,f_n\circ h).
 \label{eq:QC-control-prenaturality}
\end{align}
For \(\rho\in S_n\), let
\(\Pi_\rho:\CC[n]\to\CC[n]\) be the permutation unitary defined by
\(\Pi_\rho\lvert i\rangle=\lvert\rho(i)\rangle\), and impose
\begin{equation}
 (u_{\Pi_\rho}\otimes\id_P)
 \circ\operatorname{ctrl}^{\Gamma}_n(f_1,\ldots,f_n)
 \circ
 (u_{\Pi_{\rho^{-1}}}\otimes
  \id_{\Gamma\mathbin{\cdot}D})
 \equiv
 \operatorname{ctrl}^{\Gamma}_n
 (f_{\rho^{-1}(1)},\ldots,f_{\rho^{-1}(n)}).
 \label{eq:QC-control-permutation}
\end{equation}
For shape-compatible positive \(D,P\) and unitaries
\(U_i:\CC[I_D]\to\CC[I_P]\), also impose
\begin{equation}
 \operatorname{ctrl}_n
 (u_{U_1}^{D,P},\ldots,u_{U_n}^{D,P})
 \equiv
 u_{\bigoplus_iU_i}^{\mathbf R_n\otimes D,\,
                         \mathbf R_n\otimes P},
 \label{eq:QC-control-register-overlap}
\end{equation}
where the block sum uses the fixed lexicographic order on
\([n]\times I_D\) and \([n]\times I_P\).
\end{definition}
The last relation fixes the intended overlap between contextual
control and compound-register gates.  In particular, uniformity gives
\[
 \id_{\mathbf R_n}\otimes u_U^{D,P}
 \equiv
 u_{\bigoplus_iU}^{\mathbf R_n\otimes D,\,
                      \mathbf R_n\otimes P};
\]
this is selector--tensor compatibility, not arbitrary tensoriality of
the gate presentation.  The positive-codomain restriction keeps the
controlled output on a pure axis.

\begin{definition}[The coherent quantum core \(\QC\)]
\label{def:basic-QC}
\label{def:signed-QC}
\label{def:QC-doctrine-equations}
The objects of \(\QC\) are the raw types.  A term is \emph{raw} when
every formula in its derivation, including
every control premise and composition intermediate, is a raw type.
Writing \(\mathsf{RawTerm}(A,B)\) for these terms, for raw types
\(A,B\), \(\QC(A,B)\) is the set of \(\equiv\)-equivalence
classes in \(\mathsf{Term}(A,B)\) that contain a raw representative:
\[
 \QC(A,B)
 :=
 \bigl\{
 [t]_{\equiv}
 \bigm|
 t\in\mathsf{Term}(A,B),\qquad
 \exists r\in\mathsf{RawTerm}(A,B),\ r\equiv t
 \bigr\}.
\]
\end{definition}
Thus \(\QC\) is the image of the raw term calculus under the single
global ambient congruence, rather than all equivalence classes between
raw objects.  Raw
representatives are closed under identities, composition, tensor,
\(\parr\), and duality.  The standard unit-free derivations of
\[
 \eta_{X,A}:X\to A^*\parr(X\otimes A),
 \qquad
 \ev_{A,B}:(A^*\parr B)\otimes A\to B
\]
show that currying and uncurrying also restrict to \(\QC\).  We call
this inherited multiplicative, dual, and closed structure
unit-free \(*\)-autonomous.\footnote{Removing the units from a
\(*\)-autonomous presentation is technically subtler than merely
deleting two objects; see~\cite{HoustonHughesSchalk2005}.  The
\(*\)-autonomous apparatus supplies duality and \(\parr\) for the
polarized boundary semantics and the proof of essential unitarity.
Most of the source-level content---including qudit control, the quantum
switch, and the expressivity constructions---can instead be presented
in a linear \(\lambda\)-calculus whose type grammar omits the unit, or
categorically in the corresponding fragment of a symmetric monoidal
closed category.}
See Appendix~\ref{app:QC-finite-presentation} for the finitary
presentation and the unit-free closed maps.

Let \(\QC_+\) be the full subcategory on positive types. For fixed
\(n\), let \(\mathsf{Par}_n(\QC_+)\) be the full subcategory of
\(\QC_+^n\) on diagonal objects. Thus its objects are positive types and
\[
 \mathsf{Par}_n(\QC_+)(P,R)=\QC(P,R)^n,
\]
with identities and composition componentwise. Pointwise tensor gives
it an associative symmetric tensor, and \(S_n\) acts by permuting
components.

Equations~\eqref{eq:QC-control-uniform} and
\eqref{eq:QC-control-composition} make
\[
 \mathsf C_n:\mathsf{Par}_n(\QC_+)\longrightarrow\QC_+,
 \qquad
 P\longmapsto\mathbf R_n\otimes P,
\]
a functor sending \((f_i)_{i\in[n]}\) to
\(\operatorname{ctrl}_n(f_1,\ldots,f_n)\).
Equation~\eqref{eq:QC-control-permutation} gives its selector
covariance: for \(f_i:P\to R\), the branch permutation
\((f_i)\mapsto(f_{\rho^{-1}(i)})\) is carried to postcomposition by
\(u_{\Pi_\rho}\otimes\id_R\) and precomposition by
\(u_{\Pi_{\rho^{-1}}}\otimes\id_P\). If
\(\Delta_n\) repeats each arrow, then
\(\mathsf C_n\Delta_n=\mathbf R_n\otimes-\).
The contextual operations extend this block action to inputs
\(\Gamma\mathbin{\cdot}D\), and
Equation~\eqref{eq:QC-control-prenaturality} is their common-input
naturality. Since control retains one selector port
whereas pointwise tensor produces two, \(\mathsf C_n\) is not
tensor-preserving.

Put \(Q:=(1,0)\). Since \(\mathbf H_n=Q^{\oplus n}\), the carrier
normal form supplies, for every raw \(A\), the selector-row display
\[
 d_{n,A}:\mathbf H_n\otimes O(A)
 \xrightarrow{\cong}
 \bigoplus_{i=1}^n\bigl(Q\otimes O(A)\bigr).
\]
For positive \(P\), let \(\phi_{n,P}\) be the canonical
coefficient-one lexicographic flattening.  The recursive row convention
gives
\begin{equation}
 d_{\mathbf R_n\otimes P}
 =
 \phi_{n,P}
 \left(\bigoplus_{s=1}^n(\id_Q\otimes d_P)\right)
 d_{n,P}.
 \label{eq:QC-selector-display-compatibility}
\end{equation}

Regard \(\Perm(\CC)\) as \(*\)-autonomous via its dagger compact
structure.

\begin{definition}[Carrier assignment]
\label{def:QC-semantics}
Extend \(O\) structurally to the ambient objects and interpret the
structural arrows canonically. For an ambient term
\(f\in\mathsf{Term}(X,Y)\), write
\(\widetilde{\llbracket f\rrbracket}:O(X)\to O(Y)\) for its carrier.
For shape-compatible positive \(P,Q\) and unitary
\(U:\CC[I_P]\to\CC[I_Q]\), assign the register gates by
\begin{equation}
 d_Q\,\widetilde{\llbracket u_U^{P,Q}\rrbracket}\,d_P^{-1}
 =
 \bigl(U_{\mathbf j,\mathbf i}\,
       \id_{(k_P,0)}\bigr)_{
       \mathbf j\in I_Q,\,\mathbf i\in I_P}.
 \label{eq:QC-register-unitary-carrier}
\end{equation}
Negative gates receive the dual carrier maps.  For
\(f_i:\Gamma\mathbin{\cdot}D\to P\), set
\begin{equation}
 \widetilde{\llbracket
   \operatorname{ctrl}^{\Gamma}_n(f_1,\ldots,f_n)
 \rrbracket}
 =
 d_{n,P}^{-1}
 \left(
  \bigoplus_{i=1}^n
  \bigl(\id_Q\otimes
        \widetilde{\llbracket f_i\rrbracket}\bigr)
 \right)
 d_{n,\Gamma\mathbin{\cdot}D}.
 \label{eq:QC-control-carrier}
\end{equation}
\end{definition}

Formula~\eqref{eq:QC-control-carrier} keeps the qudit port in every
branch; only its address selects a branch. It is therefore not the
arbitrary carrier direct sum excluded by
Proposition~\ref{prop:omg-prime-h-witness}.

\begin{theorem}[Carrier functor]
\label{thm:basic-structure}
The carrier assignments on ambient terms respect \(\equiv\) and
therefore induce uniquely a functor
\[
 \llbracket-\rrbracket:\QC\longrightarrow\Perm(\CC).
\]
This functor preserves the inherited symmetric multiplicative, dual,
and closed structure.
\end{theorem}

\begin{proof}
Structural recursion defines the carrier on finite terms.  In
selector-row coordinates, uniformity, composition, prenaturality, and
permutation covariance are the corresponding block-diagonal
identities.  Equation~\eqref{eq:QC-selector-display-compatibility}
makes the register overlap the block matrix
\(\bigoplus_iU_i\).  Hence the assignment respects every generating
relation and descends uniquely through \(\equiv\).
\extendedonly{The coordinate verification is given in
Appendix~\ref{prop:QC-carrier-descent}.}
\end{proof}

For \(f=[t]_{\equiv}\), set
\[
 \llbracket f\rrbracket
 :=\widetilde{\llbracket t\rrbracket}.
\]
This is well defined by Theorem~\ref{thm:basic-structure}.

%-----------------------------------------------------------------------------
\subsection{An Applied Contextual Quantum Switch}
\label{subsec:contextual-switch-v2}
%-----------------------------------------------------------------------------

Fix a positive raw type \(D\), and put \(E:=D^*\parr D\). Define
\[
 c:=
 \ev_{D,D}\circ(\id_E\otimes\ev_{D,D}):
 E\otimes E\otimes D\longrightarrow D,
\]
and set
\[
 c_0:=c\circ(\sigma_{E,E}\otimes\id_D),
 \qquad
 c_1:=c.
\]
Thus \(c_0(f,g,x)=g(fx)\) and \(c_1(f,g,x)=f(gx)\). The applied
switch is
\begin{equation}
 \mathsf{QSwitch}^{\eta}_{D}
 :=\operatorname{ctrl}^{(E,E)}_2(c_0,c_1):
 \mathbf R_2\otimes E\otimes E\otimes D
 \longrightarrow
 \mathbf R_2\otimes D.
 \label{eq:QC-qswitch}
\end{equation}
Its curried form has type
\[
 \mathsf{QSwitch}_D:
 E\otimes E
 \longrightarrow
 \mathbf R_2^*\parr
 \bigl(D^*\parr(\mathbf R_2\otimes D)\bigr).
\]
This realizes, within the present essentially-unitary fragment, the
\(\mathtt{qcase}\) quantum switch of
Barsse, P\'echoux, and Perdrix~\cite{BarssePechouxPerdrix2026},
with an arbitrary positive payload \(D\). In selector-row coordinates,
the two diagonal blocks are the two causal orders; a unitary on
\(\mathbf R_2\) coherently superposes them.

Thus the positive-output discipline still realizes coherent control of
causal order: the shared higher-order arguments occur in both branches,
while the selector is retained in the codomain
\(\mathbf R_2\otimes D\).

%-----------------------------------------------------------------------------
\subsection{Interface with Essential Unitarity}
\label{subsec:QC-EU-interface}
%-----------------------------------------------------------------------------

For a source type \(X\), write
\(\mathsf A_X^\pm:=\mathsf A_{O(X)}^\pm\).

For \(h:A\to B\), the next section applies the qudit--MLL sequent
matrix of Definition~\ref{def:qudit-mll-sequent-matrix} to the compact
name of \(\llbracket h\rrbracket\).  The resulting carrier matrix
has type
\[
 M(h):\mathsf A_A^+\otimes\mathsf A_B^-
 \longrightarrow\mathsf A_A^-\otimes\mathsf A_B^+.
\]
The next section proves its carrier Gram identities.  The fixed-width
representation then makes the full qudit--MLL matrix dagger-unitary,
which is the source essential-unitarity condition.

\section{Essential Unitarity of the Coherent Syntax}
\label{sec:calculus}

Section~\ref{sec:eu-predicate} assigns a carrier matrix to every
nonempty qudit--MLL boundary and a full numerical matrix to every
balanced one.  A source morphism is essentially unitary when its full
matrix is dagger-unitary.  We prove this condition for every morphism
of \(\QC\).

The argument has four steps.  A source expression is translated to an
administrative expression with the same compact-name carrier point.
A categorical normalization replaces its compositions at mixed
intermediate types by carrier-equal compositions through pure
intermediates.  On pure interfaces, the qudit--MLL matrix preserves
composition and dagger.  Finally, multiplicative base maps,
positive-register gates, and contextual control contribute,
respectively, coefficient-one matching matrices, unitary qudit matrices,
and synchronized block-diagonal matrices.

Throughout, \(I:=\Unit\) denotes the tensor unit of \(\Perm(\CC)\).  Recall \(Q=(1,0)\) and \(R_r=(r,0)\) for \(r\geq0\).
For a source sequent \(\Sigma=(X_1,\ldots,X_s)\), put
\(O(\Sigma):=O(X_1)\otimes\cdots\otimes O(X_s)\), with
\(O(\varnothing):=I\).  In every
boundary datum imported from
Definition~\ref{def:qudit-mll-sequent-matrix}---namely
\(\mathsf A_\Sigma^\pm,k_\Sigma^\pm,N_\Sigma^\pm,d_\Sigma^\pm,
\zeta_\Sigma,\mathcal Q_\Sigma,\mathbf Q_\Sigma\)---the subscript
\(\Sigma\) abbreviates the carrier sequent
\((O(X_1),\ldots,O(X_s))\).
For positive \(P\), these data agree with the display
\eqref{eq:QC-positive-display}: \(\mathsf A_P^-=I\),
\(\mathsf A_P^+=O(P)\), \(d_P^+=d_P\), \(k_P^+=k_P\), and
\(N_P^+=N_P\), with the fixed unit coherence suppressed.

For \(h:A\to B\), let
\[
 a_h:=
 \operatorname{name}_{\Perm(\CC)}(\llbracket h\rrbracket):
 I\longrightarrow O(A^*,B),
 \qquad
 M(h):=\mathcal Q_{(A^*,B)}(a_h).
\]
Thus
\[
 M(h):
 \mathsf A_A^+\otimes\mathsf A_B^-
 \longrightarrow
 \mathsf A_A^-\otimes\mathsf A_B^+.
\]
If \(A\) and \(B\) are positive, unit coherence identifies \(M(h)\)
with \(\llbracket h\rrbracket\); if \(A=P^*\) and \(B=Q^*\) with
\(P,Q\) positive, it identifies \(M(h)\) with the compact transpose
\(\llbracket h\rrbracket^*:O(Q)\to O(P)\).
Put \(\operatorname{chg}(X):=k_X^+-k_X^-\).

\begin{lemma}[Boundary balance]
\label{lem:final-boundary-balance}
For every \(h:A\to B\) in \(\QC\),
\[
 k_A^+ + k_B^- = k_A^- + k_B^+ =:r_h\geq1.
\]
\end{lemma}

\begin{proof}
Choose a finite raw representative of \(h\) and induct on that term.
Unit-free multiplicative structure pairs qudit literals with their
duals, and shape-compatible register gates have equal source and target
widths.  Tensor and \(\parr\) add charges, duality reverses them, and
composition preserves equality.  In the control case, the induction
hypothesis gives
\(\operatorname{chg}(\Gamma\mathbin{\cdot}D)=\operatorname{chg}(P)\); adjoining
\(\mathbf R_n\) to both endpoints preserves it.  Hence
\(\operatorname{chg}(A)=\operatorname{chg}(B)\).

If \(r_h=0\), the displayed equality forces all four widths of \(A\)
and \(B\) to vanish, contradicting the nonempty raw-type grammar.
\end{proof}

Write \(\mathbf M(h):=\mathbf Q_{(A^*,B)}(a_h)\) for the full
qudit--MLL matrix.  By definition, \(\operatorname{EU}(h)\) means that
\(\mathbf M(h)\) is dagger-unitary.

%-----------------------------------------------------------------------------
\subsection{Matrices of the constructors}
\label{subsec:calculus-constructors}
%-----------------------------------------------------------------------------

Write \(X_\pm:=\mathsf A_X^\pm\).  A coefficient-one entry means one
occupied Kelly--Laplaza basis isomorphism with scalar \(1\).

\begin{lemma}[Multiplicative base]
\label{lem:final-multiplicative-base}
If \(p:A\to B\) is represented by a correct unit-free no-Mix
multiplicative net, then the matrix
\(M(p):A_+\otimes B_-\to A_-\otimes B_+\) is dagger-unitary.
In canonical row displays, each row and column contains exactly one
nonzero coefficient-one entry.
\end{lemma}

\begin{proof}
A correct unit-free no-Mix multiplicative net matches atomic
occurrences bijectively and preserves their qudit ranks.  At a literal
\(\mathbf R_n\), the carrier identity expands into \(n\) diagonal
rows.  Hence the whole net sends each tuple of qudit addresses to one
tuple, while its remaining KL ports carry the corresponding
coefficient-one matching.  The standard correctness, sequentialization, and cut-elimination theorem
for unit-free no-Mix MLL~\cite{GIRARD19871,Danos1989-DANTSO-5}
yields a correct cut-free presentation. Every switching of a correct
no-Mix net is connected, so no component is detached from the conclusions
and no loop scalar occurs.
Reversing the matching gives its dagger, and the row and column
equations give both Gram identities.
\end{proof}

\begin{lemma}[Register gates]
\label{lem:final-register-gates}
Let \(P,Q\) be shape-compatible positive types and let
\(U:\CC[I_P]\to\CC[I_Q]\) be unitary. Then
\[
 d_Q\,M(u_U^{P,Q})\,d_P^{-1}
 =
 \bigl(U_{\mathbf j,\mathbf i}\,
       \id_{R_{k_P}}\bigr)_{
       \mathbf j\in I_Q,\,\mathbf i\in I_P}.
\]
Consequently, \(M(u_U^{P,Q})\) and the matrix of its source dual are
dagger-unitary.
\end{lemma}

\begin{proof}
Since \(P_-=Q_-=I\), naming and polarized un-naming reduce
\(M(u_U^{P,Q})\) to its carrier gate, up to the fixed unitors. The
formula is therefore \eqref{eq:QC-register-unitary-carrier}, and its
Gram identities are those of \(U\). Compact transposition commutes
with dagger and carries both identities to the dual gate.
\end{proof}
For source morphisms \(f:A\to B\) and \(g:C\to D\), the fixed
polarized rearrangements satisfy
\begin{equation}
 M(f\diamond g)
 =
 \sigma_{\mathrm{out}}^\diamond
 \bigl(M(f)\otimes M(g)\bigr)
 \sigma_{\mathrm{in}}^\diamond,
 \qquad \diamond\in\{\otimes,\parr\}.
\label{eq:final-multiplicative-matrix}
\end{equation}
Here \(\sigma_{\mathrm{in}}^\diamond\) regroups
\((A\diamond C)_+\otimes(B\diamond D)_-\) as
\((A_+\otimes B_-)\otimes(C_+\otimes D_-)\), and
\(\sigma_{\mathrm{out}}^\diamond\) performs the inverse gathering on
the output axes.  These maps, exchange, and the corresponding axis
swaps for source duality are coefficient-one dagger-unitaries.  Hence
tensor, \(\parr\), exchange, and duality preserve dagger-unitarity of
the collected matrix.

More precisely, for \(h:A\to B\),
\begin{equation}
 M(h^*)=
 \sigma_{A_-,B_+}\,M(h)\,\sigma_{B_-,A_+}:
 B_-\otimes A_+\longrightarrow B_+\otimes A_-,
\label{eq:final-duality-matrix}
\end{equation}
where \(\sigma_{X,Y}:X\otimes Y\to Y\otimes X\).

\begin{lemma}[Contextual control]
\label{lem:final-control-matrix}
Let \(\Gamma\) be a finite raw context, let \(D,P\) be positive,
let \(f_i:\Gamma\mathbin{\cdot}D\to P\), and put
\(A:=\Gamma\mathbin{\cdot}D\) and
\(t:=\operatorname{ctrl}^{\Gamma}_n(f_1,\ldots,f_n)\).  There are fixed
coefficient-one dagger-unitary selector-row isomorphisms
\[
 \begin{aligned}
 \sigma_{\mathrm{in}}:\;&
 \mathbf H_n\otimes A_+
 \xrightarrow{\cong}\bigoplus_{i=1}^n(Q\otimes A_+),\\
 \sigma_{\mathrm{out}}:\;&
 \bigoplus_{i=1}^n(Q\otimes A_-\otimes P_+)
 \xrightarrow{\cong}A_-\otimes\mathbf H_n\otimes P_+
 \end{aligned}
\]
such that
\begin{equation}
 M(t)
 =
 \sigma_{\mathrm{out}}
 \left(
   \bigoplus_{i=1}^n
   (\id_Q\otimes M(f_i))
 \right)
 \sigma_{\mathrm{in}}.
\label{eq:final-control-matrix}
\end{equation}
If every \(M(f_i)\) is dagger-unitary, then so is \(M(t)\).
\end{lemma}

\begin{proof}
Apply polarized naming to \eqref{eq:QC-control-carrier}.  The selector
address chooses its unique branch, while the selector KL port remains
as \(\id_Q\).  Thus the middle map in
\eqref{eq:final-control-matrix} has blocks
\[
 \id_Q\otimes M(f_i):
 Q\otimes A_+
 \longrightarrow
 Q\otimes A_-\otimes P_+.
\]
The fixed maps \(\sigma_{\mathrm{in}}\) and
\(\sigma_{\mathrm{out}}\) only identify the selector rows and
reorder the polarized factors.  Orthogonality of distinct selector
rows and the two premise Gram identities prove the result.
\end{proof}

%-----------------------------------------------------------------------------
\subsection{Administrative representation}
\label{subsec:calculus-administrative}
%-----------------------------------------------------------------------------

Let \(\mathsf{Adm}_{\QC}\) be the least class of nonempty one-sided
expressions generated by the ordinary finite unit-free multiplicative
rules without Mix, together with exchange and explicit composition.
Rule contexts may be empty:
\[
 \frac{}{\vdash p^*,p}\ \mathsf{Ax},
 \qquad
 \frac{\pi:\vdash\Gamma,A\quad\theta:\vdash\Delta,B}
      {\vdash\Gamma,\Delta,A\otimes B}\ \otimes,
 \qquad
 \frac{\pi:\vdash\Gamma,A,B}
      {\vdash\Gamma,A\parr B}\ \parr,
\]
\[
 \frac{\pi:\vdash\Gamma}{\vdash\sigma\Gamma}\ \mathsf{Ex},
 \qquad
 \frac{\pi:\vdash\Gamma,D\quad
       \theta:\vdash D^*,\Delta}
      {\theta\circ_D\pi:\vdash\Gamma,\Delta}\ \mathsf{Comp}_D.
\]
Here \(p\) is a qudit literal; write \(\mathsf I_p\) for its axiom
expression. Compound identities are expanded recursively through their
outer \(\otimes\)- and \(\parr\)-nodes.

For shape-compatible positive \(P,Q\) and unitary
\(U:\CC[I_P]\to\CC[I_Q]\), add a gate coupon and its reflected
negative form:
\[
 \frac{}{\vdash P^*,Q}\ \mathsf{Gate}_{P,Q}(U),
 \qquad
 \frac{}{\vdash Q,P^*}\ \mathsf{Gate}_{P,Q}(U)^*.
\]
Their carrier points are the compact names of
\(\llbracket u_U^{P,Q}\rrbracket\) and
\(\llbracket (u_U^{P,Q})^*\rrbracket\), respectively.
The administrative control rule retains an arbitrary shared context:
\begin{equation}
 \frac{(\pi_i:\vdash\Delta,P)_{i\in[n]}}
      {\mathsf{Ctrl}_n^\Delta(\pi_1,\ldots,\pi_n):
       \vdash\Delta,\mathbf R_n^*,\mathbf R_n\otimes P}
 \qquad P\ \text{positive}.
\label{eq:final-administrative-control}
\end{equation}
The rule applies to every finite well-typed administrative family.
For source terms
\(f_i:\Gamma\mathbin{\cdot}D\to P\), take
\(\Delta=((\Gamma\mathbin{\cdot}D)^*)\), exchange
\(\mathbf R_n^*\) first, and restore the outer \(\parr\)-node:
\[
 \vdash\mathbf R_n^*\parr(\Gamma\mathbin{\cdot}D)^*,
        \mathbf R_n\otimes P
 =
 \vdash(\mathbf R_n\otimes(\Gamma\mathbin{\cdot}D))^*,
        \mathbf R_n\otimes P.
\]

Every \(\pi:\vdash\Sigma\) has a carrier point
\(a_\pi:I\to O(\Sigma)\).  The multiplicative clauses are standard.
For explicit composition, let
\[
 \alpha_{\Gamma,D,\Delta}:
 (O(\Gamma)\otimes O(D))
 \otimes(O(D^*)\otimes O(\Delta))
 \xrightarrow{\cong}
 O(\Gamma)\otimes
 (O(D)\otimes O(D)^*)\otimes O(\Delta)
\]
be the canonical coherence, using \(O(D^*)\cong O(D)^*\).  Then
\begin{equation}
 a_{\theta\circ_D\pi}
 =
 \bigl(
   \id_{O(\Gamma)}\otimes\epsilon_{O(D)}
   \otimes\id_{O(\Delta)}
 \bigr)
 \alpha_{\Gamma,D,\Delta}(a_\pi\otimes a_\theta).
\label{eq:final-administrative-composition}
\end{equation}

For control, let \(p_i:\mathbf H_n\to Q\) and
\(\iota_i:Q\to\mathbf H_n\) be the selector-row projection and
injection, and let
\[
 \chi_{\Delta,P}:
 (O(\Delta)\otimes O(P))\otimes(Q^*\otimes Q)
 \xrightarrow{\cong}
 O(\Delta)\otimes Q^*\otimes Q\otimes O(P)
\]
be the canonical symmetry.  Define
\[
 \begin{aligned}
 \beta_i
 &:=
 \bigl(
  \id_{O(\Delta)}\otimes p_i^*\otimes\iota_i
  \otimes\id_{O(P)}
 \bigr)\chi_{\Delta,P}\\
 &:
 O(\Delta)\otimes O(P)\otimes Q^*\otimes Q
 \longrightarrow
 O(\Delta)\otimes\mathbf H_n^*\otimes
 \mathbf H_n\otimes O(P).
 \end{aligned}
\]
Then
\begin{equation}
 a_{\mathsf{Ctrl}_n^\Delta(\vec\pi)}
 =
 \sum_{i=1}^n
 \beta_i\bigl(a_{\pi_i}\otimes\operatorname{name}(\id_Q)\bigr).
\label{eq:final-administrative-control-point}
\end{equation}
Gate points were specified above.  Define
\[
 M(\pi):=\mathcal Q_\Sigma(a_\pi):
 \mathsf A_\Sigma^-\longrightarrow\mathsf A_\Sigma^+.
\]
For \(\pi:\vdash\Gamma,A\) and \(\theta:\vdash\Delta,B\), let
\[
 \begin{aligned}
 \rho_{\mathrm{in}}:\;&
 \mathsf A_{\Gamma,\Delta,A\otimes B}^-
 \xrightarrow{\cong}
 \mathsf A_{\Gamma,A}^-\otimes\mathsf A_{\Delta,B}^-,
 \\
 \rho_{\mathrm{out}}:\;&
 \mathsf A_{\Gamma,A}^+\otimes\mathsf A_{\Delta,B}^+
 \xrightarrow{\cong}
 \mathsf A_{\Gamma,\Delta,A\otimes B}^+
 \end{aligned}
\]
be the coefficient-one axis coherences.  The tensor rule satisfies
\[
 M(\otimes(\pi,\theta))
 =
 \rho_{\mathrm{out}}
 (M(\pi)\otimes M(\theta))
 \rho_{\mathrm{in}}.
\]
The unary \(\parr\)-rule and exchange carry the premise matrix through
the corresponding coefficient-one axis identifications.
For \(\pi_i:\vdash\Delta,P\), let
\[
 \begin{aligned}
 s_{\mathrm{in}}^\Delta:\;&
 \mathsf A_\Delta^-\otimes\mathbf H_n
 \xrightarrow{\cong}
 \bigoplus_{i=1}^n(Q\otimes\mathsf A_\Delta^-),\\
 s_{\mathrm{out}}^\Delta:\;&
 \bigoplus_{i=1}^n
 (Q\otimes\mathsf A_\Delta^+\otimes P_+)
 \xrightarrow{\cong}
 \mathsf A_\Delta^+\otimes\mathbf H_n\otimes P_+ .
 \end{aligned}
\]
Then
\begin{equation}
 M(\mathsf{Ctrl}_n^\Delta(\vec\pi))
 =
 s_{\mathrm{out}}^\Delta
 \left(
   \bigoplus_{i=1}^n(\id_Q\otimes M(\pi_i))
 \right)
 s_{\mathrm{in}}^\Delta .
\label{eq:final-administrative-control-matrix}
\end{equation}
For \(\Delta=((\Gamma\mathbin{\cdot}D)^*)\),
\eqref{eq:final-administrative-control-point} is the compact name of
the contextual carrier map \eqref{eq:QC-control-carrier}.  The fixed
outer exchange and \(\parr\)-introduction specialize
\eqref{eq:final-administrative-control-matrix} to
\eqref{eq:final-control-matrix}.

\begin{lemma}[Administrative representation]
\label{lem:final-administrative-representation}
For every finite raw term \(t:A\to B\), there is a recursively
defined \(\mathcal T(t):\vdash A^*,B\) such that
\[
 a_{\mathcal T(t)}
 =
 \operatorname{name}(\llbracket t\rrbracket).
\]
\end{lemma}

\begin{proof}
Induct on the chosen finite raw term.  Translate the unit-free
structural clauses by the ordinary multiplicative rules and each
register gate by its coupon.  Composition becomes \(\circ_D\);
tensor, \(\parr\), exchange, and duality use their corresponding
rules.  In particular,
\[
 \mathcal T(t^*):=\mathsf{Ex}(\mathcal T(t)):\vdash B,A^*.
\]
For
\(t=\operatorname{ctrl}^{\Gamma}_n(f_1,\ldots,f_n)\), the induction
hypothesis gives
\[
 \mathcal T(f_i):
 \vdash(\Gamma\mathbin{\cdot}D)^*,P.
\]
Apply \eqref{eq:final-administrative-control} with this shared context,
then use the fixed exchange and outer \(\parr\)-introduction displayed
above.  Compact coherence and
\eqref{eq:QC-control-carrier} prove the asserted equality in every
case.
\end{proof}

%-----------------------------------------------------------------------------
\subsection{Pure-intermediate normal form}
\label{subsec:calculus-normal-form}
%-----------------------------------------------------------------------------

We first isolate the composition law used after normalization.

\begin{lemma}[Composition through a pure interface]
\label{lem:final-pure-composition}
Let
\(\pi:\vdash\Gamma,D\) and
\(\theta:\vdash D^*,\Delta\). The canonical coherence maps have types
\[
 \begin{aligned}
 \gamma_+:
 &(\mathsf A_\Gamma^+\otimes\mathsf A_D^+)
   \otimes\mathsf A_\Delta^-
 \xrightarrow{\cong}
 \mathsf A_\Gamma^+\otimes
   (\mathsf A_D^+\otimes\mathsf A_\Delta^-),\\
 \gamma_-:
 &\mathsf A_\Gamma^-\otimes
   (\mathsf A_D^-\otimes\mathsf A_\Delta^+)
 \xrightarrow{\cong}
 (\mathsf A_\Gamma^-\otimes\mathsf A_D^-)
   \otimes\mathsf A_\Delta^+.
 \end{aligned}
\]
If \(D\) is positive, then
\begin{equation}
 M(\theta\circ_D\pi)
 =
 (\id_{\mathsf A_\Gamma^+}\otimes M(\theta))
 \gamma_+
 (M(\pi)\otimes\id_{\mathsf A_\Delta^-}).
\label{eq:final-positive-composition}
\end{equation}
If \(D\) is negative, then
\begin{equation}
 M(\theta\circ_D\pi)
 =
 (M(\pi)\otimes\id_{\mathsf A_\Delta^+})
 \gamma_-
 (\id_{\mathsf A_\Gamma^-}\otimes M(\theta)).
\label{eq:final-negative-composition}
\end{equation}
Consequently, composition through a pure intermediate preserves
dagger-unitarity.
\end{lemma}

\begin{proof}
Formula~\eqref{eq:final-administrative-composition} contracts the two
occurrences of \(O(D)\).  Substitute
\(O(D)\cong D_-^*\otimes D_+\).  If \(D\) is positive, then
\(D_-=I\), and the compact snake equations give
\eqref{eq:final-positive-composition}.  If \(D\) is negative, then
\(D_+=I\), and compact transposition gives
\eqref{eq:final-negative-composition}.  The maps \(\gamma_\pm\) are
the fixed associator and inverse associator, with unitors suppressed.
When \(M(\pi)\) and \(M(\theta)\) are dagger-unitary, each
right-hand side is a composite of their dagger-unitary tensor extensions
and coefficient-one coherence.
\end{proof}

Every composition subscript below selects the contracted occurrences.
Fixed exchanges put \(D\) last in the first premise and \(D^*\) first
in the second, and a final fixed exchange restores the displayed order
of the untouched contexts.

The directed relation \(\rightsquigarrow\) is closed under
administrative rule contexts and generated by the following typed
schemata.  For a qudit literal \(p\),
\[
 \theta\circ_p\mathsf I_p\rightsquigarrow\theta,
 \qquad
 \mathsf I_p\circ_p\pi\rightsquigarrow\pi.
\]
For \(\pi:\vdash\Gamma,A\), \(\rho:\vdash\Lambda,B\),
\(\theta:\vdash A^*,B^*,\Delta\), \(\xi:\vdash\Gamma,A,B\),
\(\lambda:\vdash A^*,\Delta\), and
\(\mu:\vdash B^*,\Lambda\), the principal conversions are
\begin{align}
 \parr_{A^*,B^*}(\theta)
 \circ_{A\otimes B}\otimes(\pi,\rho)
 &\rightsquigarrow
 (\theta\circ_B\rho)\circ_A\pi,
 \label{eq:final-principal-tensor-composition}\\
 \otimes(\lambda,\mu)
 \circ_{A\parr B}\parr_{A,B}(\xi)
 &\rightsquigarrow
 \mu\circ_B(\lambda\circ_A\xi).
 \label{eq:final-principal-par-composition}
\end{align}
For \(R\in\{\otimes,\parr,\mathsf{Ex}\}\), nonprincipal naturality is
\begin{align}
 \theta\circ_D R(\ldots,\pi,\ldots)
 &\rightsquigarrow
 R(\ldots,\theta\circ_D\pi,\ldots),
 \label{eq:final-left-nonprincipal-composition}\\
 R(\ldots,\theta,\ldots)\circ_D\pi
 &\rightsquigarrow
 R(\ldots,\theta\circ_D\pi,\ldots),
 \label{eq:final-right-nonprincipal-composition}
\end{align}
when the selected \(D\), respectively \(D^*\), is auxiliary for
\(R\).  Fixed formula-tree coherences are first expanded into these
rules.

Suppose that \(E\) is pure and
\(\beta\circ_E\alpha\) is an already-normalized internal composition.
If the selected \(D\)-occurrence lies in \(\alpha\), respectively
\(\beta\), then
\begin{align}
 \theta\circ_D(\beta\circ_E\alpha)
 &\rightsquigarrow
 \beta\circ_E(\theta\circ_D\alpha),
 \label{eq:final-reassociate-left-first}\\
 \theta\circ_D(\beta\circ_E\alpha)
 &\rightsquigarrow
 (\theta\circ_D\beta)\circ_E\alpha.
 \label{eq:final-reassociate-left-second}
\end{align}
If the selected \(D^*\)-occurrence lies in \(\alpha\), respectively
\(\beta\), then
\begin{align}
 (\beta\circ_E\alpha)\circ_D\pi
 &\rightsquigarrow
 \beta\circ_E(\alpha\circ_D\pi),
 \label{eq:final-reassociate-right-first}\\
 (\beta\circ_E\alpha)\circ_D\pi
 &\rightsquigarrow
 (\beta\circ_D\pi)\circ_E\alpha.
 \label{eq:final-reassociate-right-second}
\end{align}

For \(P\) positive, \(\pi_i:\vdash\Gamma,D,P\), and
\(\theta:\vdash D^*,\Delta\), contextual-control naturality gives
\begin{equation}
 \theta\circ_D
 \mathsf{Ctrl}_n^{(\Gamma,D)}(\pi_1,\ldots,\pi_n)
 \rightsquigarrow
 \mathsf{Ctrl}_n^{(\Gamma,\Delta)}
 \bigl(
   \theta\circ_D\pi_1,\ldots,\theta\circ_D\pi_n
 \bigr).
\label{eq:final-control-composition-naturality}
\end{equation}
For \(\rho:\vdash\Gamma,D\) and
\(\lambda_i:\vdash D^*,\Delta,P\),
\begin{equation}
 \mathsf{Ctrl}_n^{(D^*,\Delta)}
   (\lambda_1,\ldots,\lambda_n)\circ_D\rho
 \rightsquigarrow
 \mathsf{Ctrl}_n^{(\Gamma,\Delta)}
 \bigl(
   \lambda_1\circ_D\rho,\ldots,\lambda_n\circ_D\rho
 \bigr).
\label{eq:final-control-composition-naturality-op}
\end{equation}
All schemata range over their typed fixed-exchange variants.  Let
\(\longleftrightarrow\) be the equivalence relation generated by
\(\rightsquigarrow\).
\begin{lemma}[Carrier soundness]
\label{lem:final-administrative-soundness}
If \(\pi\longleftrightarrow\pi'\), then
\[
 a_\pi=a_{\pi'},
 \qquad
 M(\pi)=M(\pi')
\]
in \(\Perm(\CC)\).
\end{lemma}

\begin{proof}
The first four conversion classes follow from compact coherence,
monoidality of evaluation and coevaluation, naturality, and
associativity of composition. For either control conversion, calculate
in selector-row coordinates. On both sides, row \(i\) is the same
carrier contraction through \(O(D)\), and
\[
 p_i\iota_i=\id_Q,
 \qquad
 p_i\iota_j=0\quad(i\ne j).
\]
Hence the carrier points agree entrywise. The matrix equality follows
from the definition of \(M\).
\end{proof}

For an administrative rule \(R\), put
\[
 \operatorname{ht}(R(\pi_1,\ldots,\pi_s))
 =
 \begin{cases}
  0,&s=0,\\
  1+\max_i\operatorname{ht}(\pi_i),&s>0.
 \end{cases}
\]
Fixed exchanges used only to expose and restore selected occurrences
do not contribute to \(\operatorname{ht}\); explicit exchange rules do.
After expanding the fixed multiplicative macros, let \(|D|\) be the
number of \(\otimes\)- and \(\parr\)-nodes in \(D\).

\begin{theorem}[Pure-intermediate normal form]
\label{thm:final-pure-intermediate-normal-form}
Every \(\pi:\vdash\Sigma\) in \(\mathsf{Adm}_{\QC}\) has an
administrative expression \(\pi^{\mathrm p}:\vdash\Sigma\) such that
\[
 a_\pi=a_{\pi^{\mathrm p}},
 \qquad
 M(\pi)=M(\pi^{\mathrm p}),
\]
and every explicit composition in \(\pi^{\mathrm p}\) has a positive
or negative intermediate type.
\end{theorem}

\begin{proof}
First consider a root composition \(\theta\circ_D\rho\), where
\(\rho:\vdash\Gamma,D\) and \(\theta:\vdash D^*,\Delta\), whose
premises already contain only pure-intermediate compositions.
Coherence of the structural typed-component presentation lets us
work in the fixed formula-tree coordinates above: reassociation,
symmetry, and duality introduce no additional last-rule
cases~\cite{BCST1996}.
The last-rule cases split exhaustively between the
\(*\)-autonomous structure---identity, principal
tensor--\(\parr\) decomposition, auxiliary commutation, and
reassociation---and the register-control structure: register gates,
the two control-naturality orientations, and controlled input and
output.
Use lexicographic induction on
\[
 \bigl(|D|,
       \operatorname{ht}(\rho)+\operatorname{ht}(\theta)\bigr).
\]

If \(D\) is pure, retain the composition.  Otherwise \(D\) has an
outer \(\otimes\)- or \(\parr\)-node, since every qudit literal is pure.
When the adjacent rules are principal, multiplicative decomposition
replaces the distinguished composition by compositions through proper
subformulas and lowers \(|D|\).  Otherwise naturality moves it into a
strict premise and lowers the height sum.  When the adjacent rule is an
internal composition, reassociation leaves that already-pure
composition outside and moves the distinguished mixed composition into
a strict premise, again lowering the height sum.

If the distinguished occurrence lies in the shared context of a
control rule, Equations~\eqref{eq:final-control-composition-naturality}
and~\eqref{eq:final-control-composition-naturality-op} create \(n\)
separate root-composition problems.  Each lies in a strict
premise and has smaller height; the family is finite.  A composition
through the controlled input
\(\mathbf R_n\otimes(\Gamma\mathbin{\cdot}D)\) first decomposes
through \(\mathbf R_n\) and \(\Gamma\mathbin{\cdot}D\).  The
selector is positive, and the other occurrence is then in the shared
context.  The controlled output \(\mathbf R_n\otimes P\) is positive.
Literal axioms and gate coupons
expose only pure principal formulas.  This proves the root-composition
claim.

The exchanged presentation of a source dual contains the same control
occurrence and is covered by the fixed-exchange variants above.

Now use structural induction on \(\pi\).  Normalize its premises and,
when its final rule is an explicit composition, apply the root claim.
Lemma~\ref{lem:final-administrative-soundness} preserves the carrier
point and boundary matrix throughout.
\end{proof}

%-----------------------------------------------------------------------------
\subsection{Essential unitarity}
\label{subsec:calculus-EU}
%-----------------------------------------------------------------------------

\begin{proposition}[Pure-intermediate unitarity]
\label{prop:final-pure-intermediate-unitarity}
If every explicit composition in
\(\pi:\vdash\Sigma\) has a pure intermediate type, then
\[
 M(\pi)^\dagger M(\pi)=\id_{\mathsf A_\Sigma^-},
 \qquad
 M(\pi)M(\pi)^\dagger=\id_{\mathsf A_\Sigma^+}.
\]
\end{proposition}

\begin{proof}
Induct on the administrative expression.  In the qudit-row display, an
atomic axiom is diagonal, with one coefficient-one identity in every
row.  The administrative tensor, \(\parr\), and exchange rules act by
the sequent-axis formulas displayed above.  Shape-compatible gate coupons are covered by
Lemma~\ref{lem:final-register-gates}, and reflected coupons by duality.
The administrative control rule has the block form
\eqref{eq:final-administrative-control-matrix}.  Source duals introduce
only the fixed exchange and obey \eqref{eq:final-duality-matrix}.
At explicit composition, the
intermediate is pure, so Lemma~\ref{lem:final-pure-composition} applies
to the two premise matrices.  These are all administrative rules.
\end{proof}

\begin{theorem}[Essential unitarity of \(\QC\)]
\label{thm:calculus-EU}
For every \(h:A\to B\) in \(\QC\), the polarized carrier matrix
\(M(h)\) depends only on \(h\) and satisfies
\[
 M(h)^\dagger M(h)=\id_{A_+\otimes B_-},
 \qquad
 M(h)M(h)^\dagger=\id_{A_-\otimes B_+}.
\]
Consequently,
\[
 \mathbf M(h)^\dagger\mathbf M(h)
 =\id_{\CC[N_{(A^*,B)}^-]\otimes\CC[r_h]},
 \qquad
 \mathbf M(h)\mathbf M(h)^\dagger
 =\id_{\CC[N_{(A^*,B)}^+]\otimes\CC[r_h]},
\]
so \(\operatorname{EU}(h)\).
\end{theorem}

\begin{proof}
The carrier interpretation respects the equations of \(\QC\), and
compact coherence makes naming and un-naming canonical. Hence \(M(h)\)
depends only on \(h\).
By Definition~\ref{def:basic-QC}, choose a finite raw term \(t\)
with \([t]_{\equiv}=h\), and put \(\pi_h:=\mathcal T(t)\).
Administrative representation gives
\[
 a_{\pi_h}
 =\operatorname{name}(\llbracket h\rrbracket),
 \qquad
 M(h)=M(\pi_h).
\]
Theorem~\ref{thm:final-pure-intermediate-normal-form} supplies a
carrier-equivalent \(\pi_h^{\mathrm p}\) with only pure-intermediate
compositions.  Hence
\[
 M(h)=M(\pi_h)=M(\pi_h^{\mathrm p}).
\]
Proposition~\ref{prop:final-pure-intermediate-unitarity} gives both
carrier Gram identities.  Proposition~\ref{prop:first-order-representation}
carries them to the two full-matrix identities above, which are
precisely the source essential-unitarity condition.
\end{proof}

%-----------------------------------------------------------------------------
\section{One-Slot Supermaps}
\label{sec:expr-v2}\label{sec:expressiveness}
%-----------------------------------------------------------------------------

\subsection{Resource-Linear Contexts}
\label{subsec:function-types-v2}

A one-slot context uses its argument once along every selected control
branch. Composition and tensor place fixed programs around that slot;
contextual control chooses among a synchronized family without copying
its input.

Fix a slot \(\Box:A\to B\) and a class
\(\mathsf{Const}\subseteq\operatorname{Mor}(\QC)\) of well-typed constant
programs. Resource-linear contexts are generated by
\[
\begin{aligned}
 C[\Box] ::={}&\Box
 \mid C[\Box]\circ K
 \mid K\circ C[\Box] \\
 &\mid C[\Box]\otimes K
 \mid K\otimes C[\Box]
 \mid \operatorname{ctrl}^{\Gamma}_n(D_1[\Box],\ldots,D_n[\Box]),
\end{aligned}
\]
where \(K\in\mathsf{Const}\) and every expression is well typed. In the
control clause, \(E=\Gamma\mathbin{\cdot}D\) for a positive \(D\); the
contexts \(D_i[\Box]:E\to P\) have a common positive
codomain \(P\), and the resulting context has type
\(\mathbf R_n\otimes E\to\mathbf R_n\otimes P\).
The grammar is closed under the typed currying and uncurrying
bijections of \(\QC\), each preserving its single slot.

\begin{theorem}[Resource-linear contexts]
\label{thm:resource-linear-supermaps-v2}
\label{thm:context-supermap-nholes-objectplus}
Every well-typed resource-linear context
\(C[\Box]:E\to F\), with \(\Box:A\to B\), defines
\(\Phi_C:\QC(A,B)\to\QC(E,F)\), \(\Phi_C(f):=C[f]\), and admits a
\(\CC\)-linear carrier map
\[
 \widetilde\Phi_C:
 \Perm(\CC)\bigl(O(A),O(B)\bigr)
 \longrightarrow
 \Perm(\CC)\bigl(O(E),O(F)\bigr)
\]
satisfying
\(\widetilde\Phi_C(\llbracket f\rrbracket)
=\llbracket\Phi_C(f)\rrbracket\) for every \(f:A\to B\) in \(\QC\).
Every \(\Phi_C(f)\) is essentially unitary.
\end{theorem}

\begin{proof}
Induct on the context. The slot gives \(\widetilde\Phi_\Box=\id\).
Writing \(k:=\llbracket K\rrbracket\), the four multiplicative clauses
send \(x\) to \(\widetilde\Phi_D(x)\circ k\),
\(k\circ\widetilde\Phi_D(x)\),
\(\widetilde\Phi_D(x)\otimes k\), and
\(k\otimes\widetilde\Phi_D(x)\), respectively.
For a control clause with \(D_i[\Box]:E\to P\), set
\[
 \widetilde\Phi_C(x)
 :=d_{n,P}^{-1}
 \left(
  \bigoplus_{i=1}^n
  \bigl(\id_Q\otimes\widetilde\Phi_{D_i}(x)\bigr)
 \right)d_{n,E}.
\]
These maps are linear by bilinearity and finite biproduct assembly;
currying and uncurrying use the linear compact name and unname
isomorphisms. The carrier equations give the substitution identity.
Closure of \(\QC\) and Theorem~\ref{thm:calculus-EU} give the remaining
claims.
\end{proof}

%-----------------------------------------------------------------------------
\subsection{Equal-Ratio Unitary Dilation Stages}
\label{subsec:pure-comb-v2}
%-----------------------------------------------------------------------------

The source calculus realizes the two unitary stages of an exposed-memory
one-slot dilation. Memory initialization and partial trace remain in
the external Hilbert-space/channel semantics.

For a register \(X=\mathbf R_d\), write \(H_X\cong\CC^d\) for its
external Hilbert space. For finite-dimensional \(H\), put
\(\mathcal B(H):=\operatorname{End}(H)\). We use the Schr\"odinger
convention: channels are completely positive trace-preserving maps
\(\mathcal B(H_{\mathrm{in}})\to\mathcal B(H_{\mathrm{out}})\). A
deterministic one-slot supermap is a completely CP-preserving linear map
on quantum operations that sends channels to
channels~\cite{Chiribella2008Supermaps}.

\begin{theorem}[Internal realization of equal-ratio one-slot supermaps]
\label{thm:equal-ratio-supermaps-v2}
\label{thm:equal-dim-supermaps}
Let \(A_i=\mathbf R_{a_i}\) and \(B_i=\mathbf R_{b_i}\) for
\(i=0,1\), with \(a_0b_1=a_1b_0\), and let \(\Phi\) be a deterministic one-slot
supermap from channels
\(\mathcal B(H_{A_0})\to\mathcal B(H_{A_1})\) to channels
\(\mathcal B(H_{B_0})\to\mathcal B(H_{B_1})\). Then, for some
registers \(M=\mathbf R_m\) and \(N=\mathbf R_n\), there are morphisms
in \(\QC\)
\[
 V_1^{\QC}:B_0\otimes M\longrightarrow A_0\otimes N,
 \qquad
 V_2^{\QC}:A_1\otimes N\longrightarrow B_1\otimes M
\]
whose carriers are unitary extensions of the two isometries in a
Chiribella realization of \(\Phi\). For every \(f:A_0\to A_1\) in
\(\QC\),
\[
 \Phi^\sharp(f):=
 V_2^{\QC}\circ(f\otimes\id_N)\circ V_1^{\QC}:
 B_0\otimes M\longrightarrow B_1\otimes M
\]
is a morphism of \(\QC\) and is essentially unitary.
\end{theorem}

\begin{proof}
By \cite[Theorem~1, Eq.~(23)]{Chiribella2008Supermaps}, choose
isometries
\[
 V:H_{B_0}\longrightarrow H_{A_0}\otimes H_{N_0},
 \qquad
 W:H_{A_1}\otimes H_{N_0}
   \longrightarrow H_{B_1}\otimes H_{M_0}
\]
representing \(\Phi\).
Write
\(b_0/a_0=b_1/a_1=r/s\) in lowest terms. Choose \(L\) so that
\(rL\geq\dim H_{N_0}\) and \(sL\geq\dim H_{M_0}\), and put
\(n:=rL\), \(m:=sL\). Then \(b_0m=a_0n\) and \(a_1n=b_1m\).
Choose isometric embeddings \(j_M:H_{M_0}\hookrightarrow H_M\),
\(j_N:H_{N_0}\hookrightarrow H_N\), and a unit vector
\(\ket{0}_M\in H_M\). The
isometries
\[
 \psi\otimes\ket{0}_M
 \longmapsto(\id_{H_{A_0}}\otimes j_N)V\psi,
 \qquad
 (\id_{H_{A_1}}\otimes j_N)\xi
 \longmapsto(\id_{H_{B_1}}\otimes j_M)W\xi
\]
extend to unitaries \(\widetilde V_1\) and \(\widetilde V_2\) on the
equal-dimensional ambient spaces.

Both source--target pairs are positive, have port width \(2\), and
have equal address cardinality. If \(G_1,G_2\) are the matrices of
\(\widetilde V_1,\widetilde V_2\) in the lexicographic qudit bases,
Definition~\ref{def:QC-register-unitaries} supplies
\[
 V_1^{\QC}:=u_{G_1}^{B_0\otimes M,\,A_0\otimes N},
 \qquad
 V_2^{\QC}:=u_{G_2}^{A_1\otimes N,\,B_1\otimes M}.
\]
Their carriers are the required unitary extensions. Closure of \(\QC\)
and Theorem~\ref{thm:calculus-EU} give the conclusion.
\end{proof}

Thus the positive-interface discipline realizes the unitary stages of
every equal-ratio one-slot supermap covered by the theorem.

The stage construction does not require equal endpoint dimensions. By
Theorem~\ref{thm:calculus-EU}, however, \(\QC(A_0,A_1)\ne\varnothing\)
forces \(a_0=a_1\); the equal-ratio condition then forces \(b_0=b_1\).

For \(g:P\to Q\) in \(\QC\) between positive types, boundary balance
gives \(k_P=k_Q=:k\). Define its full fixed-width matrix by
\[
 \widehat U_g
 :=\mathsf B_k\!\left(
 d_Q\,\llbracket g\rrbracket\,d_P^{-1}
 \right):
 \CC[I_P]\otimes\CC[k]\longrightarrow
 \CC[I_Q]\otimes\CC[k].
\]
Proposition~\ref{prop:first-order-representation} gives
\(\widehat U_{h\circ g}=\widehat U_h\widehat U_g\), and
Theorem~\ref{thm:calculus-EU} makes \(\widehat U_g\) unitary. For
\(k=1\), identify \(\widehat U_g\) with \(U_g:\CC[I_P]\to\CC[I_Q]\).
For a unitary \(U\), write
\(\operatorname{Ad}_U(\rho):=U\rho U^\dagger\).

\begin{corollary}[External readout of the internal lift]
\label{cor:equal-ratio-supermap-readout}
For the realization of Theorem~\ref{thm:equal-ratio-supermaps-v2}
and every \(f:A_0\to A_1\) in \(\QC\), the full matrix factors as
\[
 \widehat U_{\Phi^\sharp(f)}
 =U_{\Phi^\sharp(f)}\otimes\id_{\CC[2]},
 \qquad
 U_{\Phi^\sharp(f)}
 =\widetilde V_2(U_f\otimes\id_{H_N})\widetilde V_1.
\]
There is a pure state \(\omega_M\) on \(H_M\), independent of \(f\),
such that, for every such \(f\) and every density operator \(\rho\) on
\(H_{B_0}\),
\[
 \Phi(\operatorname{Ad}_{U_f})(\rho)
 =\operatorname{Tr}_{H_M}\!\left[
 \operatorname{Ad}_{U_{\Phi^\sharp(f)}}(\rho\otimes\omega_M)
 \right].
\]
\end{corollary}

\begin{proof}
Use the embeddings and unitary extensions constructed in the proof of
Theorem~\ref{thm:equal-ratio-supermaps-v2}, and put
\(\omega_M:=\ket{0}_M\!\bra{0}_M\). For every channel
\(\mathcal F:\mathcal B(H_{A_0})\to\mathcal B(H_{A_1})\), the
extension identities give
\[
 \Phi(\mathcal F)(\rho)
 =\operatorname{Tr}_{H_M}\!\left[
 \widetilde V_2
 (\mathcal F\otimes\operatorname{id}_{\mathcal B(H_N)})
 \bigl(
  \widetilde V_1(\rho\otimes\omega_M)\widetilde V_1^\dagger
 \bigr)
 \widetilde V_2^\dagger
 \right].
\]
Here \(\mathcal F\otimes\operatorname{id}\) preserves the
\(j_N\)-supported block; the defining extension identity for
\(\widetilde V_2\) then produces the \(j_M\)-embedded sequential output,
and \(\operatorname{Tr}_{H_M}\) removes this final isometric embedding
\cite[Eq.~(23)]{Chiribella2008Supermaps}. Taking
\(\mathcal F=\operatorname{Ad}_{U_f}\), the register-gate carrier
equation, monoidality, and
Proposition~\ref{prop:first-order-representation} give the displayed
factorization of \(\widehat U_{\Phi^\sharp(f)}\), hence the result.
\end{proof}

%============================================================
% Conclusion
%============================================================
\section{Conclusion}

We asked what boundary predicate extends ordinary unitarity to
higher-order typed interfaces. For every arrow \(h:A\to B\) of
\(\QC\), compact naming and polarized un-naming give
\[
 M(h):
 \mathsf A_A^+\otimes\mathsf A_B^-
 \longrightarrow
 \mathsf A_A^-\otimes\mathsf A_B^+.
\]
The proof establishes
\[
 M(h)^\dagger M(h)
 =
 \id_{\mathsf A_A^+\otimes\mathsf A_B^-},
 \qquad
 M(h)M(h)^\dagger
 =
 \id_{\mathsf A_A^-\otimes\mathsf A_B^+},
\]
and the fixed-width representation transfers these identities to the
full numerical qudit--MLL matrix. Thus every \(\QC\)-morphism is
essentially unitary. The argument isolates the source discipline that
makes this possible: multiplicative structure, positive-register
gates, and contextual qudit control reduce to dagger-unitary matrices
composed through pure interfaces.

The applied contextual quantum switch illustrates coherent control at
a higher-order interface. Resource-linear one-slot contexts remain in
\(\QC\), and equal-ratio one-slot dilations admit internal lifts of
their two unitary stages while the memory wire remains explicit. An
internal treatment of measurement, partial trace, and mixed-state
semantics is left to future work.
\bibliographystyle{eptcs}
\bibliography{merged}

%============================================================
% Appendices — deferred proofs (extended version only)
%============================================================
% The conference build drops the appendix; conference-gated text
% must keep statements self-contained.

\ifconfversion\else
\appendix
%-----------------------------------------------------------------------------
% APPFINAL.tex -- proofs retained by structural.tex, formal.tex, and QCFINAL.tex
%-----------------------------------------------------------------------------

%-----------------------------------------------------------------------------
\section{Foundations of the coherent syntax}
\label{app:QC-finite-presentation}
%-----------------------------------------------------------------------------

For a finite term \(t\), define \(\operatorname{depth}(t)\) by
assigning depth \(0\) to control-free terms, taking the maximum over
the premises of every structural constructor, and setting
\[
 \operatorname{depth}
 \bigl(\operatorname{ctrl}^{\Gamma}_n(f_1,\ldots,f_n)\bigr)
 =
 1+\max_{i\in[n]}\operatorname{depth}(f_i).
\]

\begin{proposition}[Finitary presentation]
\label{prop:QC-finitary-presentation}
Every \(\equiv\)-class of ambient terms has a finite representative
of finite control depth, and every morphism of \(\QC\) has a finite
raw representative.  An interpretation of the register gates and
contextual controls in a symmetric unital \(*\)-autonomous category
that satisfies the defining relations extends uniquely to a functor
on the ambient quotient, and restricts to a functor from \(\QC\).
\end{proposition}

\begin{proof}
The grammar is inductive and finitary: every structural constructor
has finite arity and every control node has finitely many proper
premises.  Thus every term has finite depth.

Structural recursion defines the image of a finite term in any target
interpretation.  Equality of interpreted maps is a typed congruence
closed under contextual control.  Since it contains every generating
relation, leastness of \(\equiv\) makes the interpretation independent
of representatives.  The same recursion proves uniqueness.  The raw
assertion is Definition~\ref{def:basic-QC}.
\end{proof}

\begin{lemma}[Unit-free closed maps]
\label{lem:QC-unit-free-closed}
For raw \(X,A,B\), the canonical maps
\[
 \eta_{X,A}:X\longrightarrow A^*\parr(X\otimes A),
 \qquad
 \ev_{A,B}:(A^*\parr B)\otimes A\longrightarrow B
\]
have raw representatives.
\end{lemma}

\begin{proof}
Using recursively expanded identities, the one-sided unit-free
multiplicative derivation
\[
 \frac{
   \dfrac{
     \vdash X^*,X
     \qquad
     \vdash A^*,A
   }{
     \vdash X^*,A^*,X\otimes A
   }\ \otimes
 }{
   \vdash X^*,A^*\parr(X\otimes A)
 }\ \parr
\]
represents \(\eta_{X,A}\).  For evaluation, tensor the identity
derivations with \(A\) and \(B^*\) principal:
\[
 \frac{
   \vdash A^*,A
   \qquad
   \vdash B,B^*
 }{
   \vdash A^*,B,A\otimes B^*
 }\ \otimes .
\]
After fixed exchange and par introduction this becomes
\[
 \vdash (A\otimes B^*)\parr A^*,B.
\]
Since
\[
 ((A^*\parr B)\otimes A)^*
 =(A\otimes B^*)\parr A^*,
\]
it is the compact-name sequent for \(\ev_{A,B}\).  Every formula in
both derivations is raw.

For raw \(f:X\otimes A\to B\) and
\(g:X\to A^*\parr B\), the standard terms
\[
 \operatorname{curry}(f)
 = (\id_{A^*}\parr f)\circ\eta_{X,A},
 \qquad
 \operatorname{uncurry}(g)
 = \ev_{A,B}\circ(g\otimes\id_A)
\]
have raw intermediates.  The ambient \(\beta\)- and \(\eta\)-equations
therefore restrict the currying bijection to \(\QC\).
\end{proof}

\begin{proposition}[Carrier descent]
\label{prop:QC-carrier-descent}
The carrier assignment of Definition~\ref{def:QC-semantics} respects
\(\equiv\).
\end{proposition}

\begin{proof}
For carrier maps \(F_i:O(A)\to O(B)\), set
\[
 B_n(F_1,\ldots,F_n)
 :=
 \bigoplus_{i=1}^n(\id_Q\otimes F_i):
 \bigoplus_i(Q\otimes O(A))
 \longrightarrow
 \bigoplus_i(Q\otimes O(B)).
\]
For \(G_i:O(B)\to O(C)\) and \(H:O(A')\to O(A)\),
\[
 \begin{aligned}
 B_n(F,\ldots,F)
 &=d_{n,B}(\id_{\mathbf H_n}\otimes F)d_{n,A}^{-1},\\
 B_n(G_1,\ldots,G_n)B_n(F_1,\ldots,F_n)
 &=B_n(G_1F_1,\ldots,G_nF_n),\\
 B_n(F_1,\ldots,F_n)B_n(H,\ldots,H)
 &=B_n(F_1H,\ldots,F_nH).
 \end{aligned}
\]
For \(\rho\in S_n\), put
\[
 S_{\rho,A}
 :=
 d_{n,A}
 (\widetilde{\llbracket u_{\Pi_\rho}\rrbracket}
  \otimes\id_{O(A)})
 d_{n,A}^{-1}.
\]
Then
\[
 S_{\rho,B}B_n(F_1,\ldots,F_n)S_{\rho^{-1},A}
 =
 B_n(F_{\rho^{-1}(1)},\ldots,F_{\rho^{-1}(n)}).
\]
Together with
\[
 \widetilde{\llbracket
   \operatorname{ctrl}^{\Gamma}_n(\vec f)
 \rrbracket}
 =
 d_{n,P}^{-1}
 B_n(\widetilde{\llbracket f_1\rrbracket},\ldots,
     \widetilde{\llbracket f_n\rrbracket})
 d_{n,\Gamma\mathbin{\cdot}D},
\]
these identities verify uniformity, composition, prenaturality, and
selector covariance.

It remains to check the register overlap.  Let
\[
 W:=\bigoplus_{s=1}^nU_s,
 \qquad
 W_{(s,\mathbf j),(t,\mathbf i)}
 =\delta_{s,t}(U_s)_{\mathbf j,\mathbf i}.
\]
Using Equation~\eqref{eq:QC-selector-display-compatibility} for
\(D\) and \(P\), the fully displayed carrier of the controlled
branch gates is
\[
 \left(
 W_{(s,\mathbf j),(t,\mathbf i)}
 \id_{(1+k_D,0)}
 \right)_{
 (s,\mathbf j)\in[n]\times I_P,\quad
 (t,\mathbf i)\in[n]\times I_D}.
\]
Equation~\eqref{eq:QC-register-unitary-carrier} assigns the same
carrier to \(u_W^{\mathbf R_n\otimes D,\mathbf R_n\otimes P}\), so
the overlap relation is respected.  The structural and register
relations hold by construction, completing the descent through the
global congruence.
\end{proof}

\section{The loop-erased carrier}
\label{app:structural-proofs}

\begin{proof}[Proof of Proposition~\ref{prop:well-defined-assoc}]
The notation \(\mathcal L\), the operations \(g\star f\) and
\(\ell(g,f)\), and the basis composition rule are those of
Definition~\ref{def:perm-C}.

\paragraph*{Identities and associativity.}
Let \(1_M\) be the straight loop-free identity matching. Gluing it to
either side of \(f\) gives the external paths of \(f\) and no closed
component:
\[
 f\star1_M=f,
 \qquad
 1_N\star f=f,
 \qquad
 \ell(f,1_M)=\ell(1_N,f)=0.
\]
Thus \([1_M]\) is the identity for
Equation~\eqref{eq:delta-twisted-basis-composition}.

For \(f:M\to N\), \(g:N\to P\), and \(h:P\to Q\), identify both
intermediate boundaries at once and call the resulting graph
\(\Gamma(h,g,f)\). Its external paths are independent of the order of
the two identifications, so
\[
 h\star(g\star f)=(h\star g)\star f.
\]
The closed components are independent of that order as well. The two
possible partitions of those components give
\[
 \ell(g,f)+\ell(h,g\star f)
 =
 \ell(h,g)+\ell(h\star g,f).
\]
Consequently
\[
\begin{aligned}
 [h]\circ([g]\circ[f])
 &=
 \delta^{\ell(g,f)+\ell(h,g\star f)}
 [h\star(g\star f)]                                      \\
 &=
 \delta^{\ell(h,g)+\ell(h\star g,f)}
 [(h\star g)\star f]
 =([h]\circ[g])\circ[f].
\end{aligned}
\]
Bilinearity proves associativity for arbitrary elements of
\(\mathcal L\). For composable matrices \(F,G,H\), finite
distributivity and the basis equation give
\[
 \bigl(H\circ(G\circ F)\bigr)_{qi}
 =
 \sum_{k,j}H_{qk}\circ(G_{kj}\circ F_{ji})
 =
 \sum_{k,j}(H_{qk}\circ G_{kj})\circ F_{ji}
 =
 \bigl((H\circ G)\circ F\bigr)_{qi}.
\]
The diagonal matrices \(\operatorname{diag}([1_{M_i}])\) are the
matrix identities.

\paragraph*{Dagger and tensor.}
Reversing every path in \(\Gamma(g,f)\) preserves its closed
components and reverses its external matching. Since
\(\bar\delta=\delta\),
\[
\begin{aligned}
 \ell(f^\dagger,g^\dagger)&=\ell(g,f),
 &
 f^\dagger\star g^\dagger&=(g\star f)^\dagger,\\
 ([g]\circ[f])^\dagger
 &=\delta^{\ell(g,f)}[(g\star f)^\dagger]
 =[f]^\dagger\circ[g]^\dagger.
\end{aligned}
\]
Entrywise expansion gives
\[
 \bigl((G\circ F)^\dagger\bigr)_{ik}
 =
 \sum_jF_{ji}^\dagger\circ G_{kj}^\dagger
 =
 (F^\dagger\circ G^\dagger)_{ik}.
\]

Disjoint union gives tensor on basis matchings. For two composable
pairs,
\[
\begin{aligned}
 (g_1\otimes g_2)\star(f_1\otimes f_2)
 &=(g_1\star f_1)\otimes(g_2\star f_2),\\
 \ell(g_1\otimes g_2,f_1\otimes f_2)
 &=\ell(g_1,f_1)+\ell(g_2,f_2).
\end{aligned}
\]
Substitution in Equation~\eqref{eq:delta-twisted-basis-composition}
proves tensor interchange on basis morphisms. On finite families,
tensor is the Kronecker matrix
\[
 (F\otimes F')_{(j,j'),(i,i')}
 =F_{ji}\otimes F'_{j'i'},
\]
so finite distributivity gives
\[
 (G\otimes G')\circ(F\otimes F')
 =
 (G\circ F)\otimes(G'\circ F').
\]
The associator, unitors, and symmetry are the corresponding loop-free
reindexing matrices. Both sides of each coherence equation have the
same external paths and create no closed component.

\paragraph*{Compact structure.}
For a signed interface \(M\), take the ordinary loop-free KL cup and
cap
\[
 [\eta_M]:\Unit\longrightarrow M^*\otimes M,
 \qquad
 [\epsilon_M]:M\otimes M^*\longrightarrow\Unit.
\]
In either yanking composite, every port belongs to an external path and
no closed component is formed. Hence
\[
 ([\epsilon_M]\otimes1_M)
 \circ(1_M\otimes[\eta_M])=1_M,
 \qquad
 (1_{M^*}\otimes[\epsilon_M])
 \circ([\eta_M]\otimes1_{M^*})=1_{M^*}.
\]
If \(\sigma_{M,M^*}:M\otimes M^*\to M^*\otimes M\) is the
loop-free symmetry matching, reversal gives
\[
 [\epsilon_M]
 =
 [\eta_M]^\dagger\circ[\sigma_{M,M^*}].
\]

For a finite family \(A=\bigoplus_iM_i\), the rows of
\(A^*\otimes A\) are indexed by \((i,j)\). Define \(\eta_A\) to have
the single entry \([\eta_{M_i}]\) in each diagonal row \((i,i)\) and
zero in every off-diagonal row; define \(\epsilon_A\) by the
transposed diagonal matrix of caps. Equation
\eqref{eq:delta-twisted-matrix-composition} makes each yanking
composite the diagonal identity matrix: its \(ii\)-entry is
\(1_{M_i}\), and every \(ij\)-entry with \(i\ne j\) is zero. Thus
\(\Perm(\Bbbk)\) is dagger compact closed. Its empty family is the zero
object, and the standard block injections and projections give finite
biproducts.

\paragraph*{Linearization and trace.}
Write a KL morphism as \((\tau,\kappa)\), where \(\tau\) is its
loop-free matching. If \(f=(\tau,\kappa)\) and
\(g=(\sigma,\lambda)\), KL execution gives
\[
 g\circ_{\mathbf{KL}}f
 =
 (\sigma\star\tau,\kappa+\lambda+\ell(\sigma,\tau)).
\]
With \(L(\tau,\kappa):=\delta^\kappa[\tau]\), the twist records
exactly the loop count, so
\[
 L(g)\circ L(f)
 =
 \delta^{\lambda+\kappa+\ell(\sigma,\tau)}
 [\sigma\star\tau]
 =L(g\circ_{\mathbf{KL}}f).
\]
The identity, dagger, and tensor equations above also commute with
\(L\). Hence \(L:\mathbf{KL}\to\mathcal L\) is a dagger monoidal
functor.

The cups and caps supply the canonical trace of a compact closed
category \cite{joyalStreetVerity1996}. On a basis matching, this trace
performs the same path execution and contributes one factor of
\(\delta\) for each closed component; it therefore agrees with
Equation~\eqref{eq:delta-twisted-basis-composition}. The matrix trace is
obtained from the displayed block cups and caps.

\end{proof}

%-----------------------------------------------------------------------------
\section{Boundary execution}
\label{app:boundary-proofs}
%-----------------------------------------------------------------------------

\begin{proof}[Proof of Proposition~\ref{prop:boundary-execution-gluing}]
For \(f:M\to N\), \(g:N\to P\), and
\(x\in\Pos(M)\sqcup\Neg(P)\), put \(x_0=x\) and follow the unique
alternating path
\[
 x_{k+1}=
 \begin{cases}
  \tau_f(x_k),&x_k\in\Pos(M)\sqcup\Neg(N),\\
  \tau_g(x_k),&x_k\in\Pos(N)\sqcup\Neg(P).
 \end{cases}
\]
Stop at the first \(x_m\in\Neg(M)\sqcup\Pos(P)\). The external
component of the glued graph satisfies
\[
 \tau_{g\circ f}(x)=x_m.
\]
Every remaining component lies wholly on the identified \(N\)-ports
and is a cycle. If their number is \(\ell(g,f)\), then
\[
 \kappa(g\circ f)
 =
 \kappa(f)+\kappa(g)+\ell(g,f).
\]
This is the paths-and-cycles decomposition of
\cite[Proposition~1]{AbramskyAS}.
\end{proof}

%-----------------------------------------------------------------------------
\section{The motivating witness calculations}
\label{app:QC-h-witness}
%-----------------------------------------------------------------------------

\begin{proof}[Proof of Proposition~\ref{prop:omg-prime-h-witness}]
Use the displayed port labels. In the composite of two copies, mark the
final target by double primes. Kelly--Laplaza execution gives
\[
 a\longmapsto b'',
 \qquad
 b\longmapsto c,
 \qquad
 c''\longmapsto a''.
\]
After renaming the final target, these are precisely the three links of
\(\bar w\); hence \(\bar w^2=\bar w\). This matching is not the identity, so
\(\bar w\ne\id_M\); an invertible idempotent is the identity, and hence
\(\bar w\) is not invertible.

Put \(L:=\id_M\oplus\bar w\) and \(K:=LXL\). In the fixed
qudit--MLL display used to define \(\mathsf C\), the six execution paths
are
\[
\begin{array}{c|c|c}
 x&\mathsf C(L)(x)&\mathsf C(K)(x)\\ \hline
 a_0&a'_0&b'_1\\
 b_0&b'_0&c_1\\
 c'_0&c_0&a'_1\\
 a_1&b'_1&b'_0\\
 b_1&c_1&c_0\\
 c'_1&a'_1&a'_0.
\end{array}
\]
Therefore
\[
\begin{aligned}
 \mathsf C(L)^{-1}\mathsf C(K):\quad
 &a_0\longmapsto a_1\longmapsto b_0\longmapsto b_1
       \longmapsto c'_0\longmapsto c'_1\longmapsto a_0.
\end{aligned}
\]
and hence
\[
 R:=\mathsf C(L)^{-1}\mathsf C(K)
   =(a_0\ a_1\ b_0\ b_1\ c'_0\ c'_1).
\]
The matrix \(\mathsf C(L)\) is a permutation.
In selector-row coordinates,
\[
 S:=\mathsf C(\id)^{-1}\mathsf C(X)
\]
is the qudit summand swap and hence a Hermitian involution. Therefore
\[
 \mathsf C(\cos\theta\,\id+\ii\sin\theta\,X)
 =\mathsf C(\id)(\cos\theta\,I+\ii\sin\theta\,S)
\]
is unitary, as is \(\mathsf C(L)\).
Since \(L^2=L\), linearity of the qudit--MLL matrix gives
\[
\begin{aligned}
 \mathsf C(h)
 &=\cos\theta\,\mathsf C(L)
   +\ii\sin\theta\,\mathsf C(K)\\
 &=\mathsf C(L)
   (\cos\theta\,I+\ii\sin\theta\,R),\\
 \mathsf C(h)^\dagger\mathsf C(h)
 &=I+\ii\cos\theta\sin\theta\,(R-R^{-1}).
\end{aligned}
\]
The permutation \(R\) is a six-cycle, so \(R\ne R^{-1}\). The last
Gram matrix is therefore not the identity when
\(\theta\notin\frac{\pi}{2}\mathbb Z\).
\end{proof}

\begin{proof}[Proof of Proposition~\ref{prop:QC-F2-witness}]
Let \(c_1,c_2\) be the two target-\(Y\) ports. Let \(x_1,x_3\) denote
the \(P\)-occurrences of \(Z\), and \(x_2,x_4\) its
\(P^*\)-occurrences. Take \(\bar q:Z\to Y\) with axiom links
\[
 x_1\text{--}c_1,
 \qquad
 x_4\text{--}c_2,
 \qquad
 x_2\text{--}x_3
 \quad\text{(cross-half cap)}.
\]
In the boundary matrix \(\mathsf B_{Z,Y}\), the
ratio of the matchings of \(\bar q\) and \(\bar q\sigma\) is the
three-cycle
\[
 t:=\tau_{\bar q\sigma}\tau_{\bar q}^{-1}
   =(c_1\ x_4\ x_2),
 \qquad
 t^\dagger=t^{-1}\ne t.
\]
The columns \(x_1\) and \(x_3\) of
\[
 \cos\theta\,\mathsf B_{Z,Y}(\bar q)
 +\ii\sin\theta\,\mathsf B_{Z,Y}(\bar q\sigma)
\]
have inner product \(-\ii\sin\theta\cos\theta\). It is nonzero when
\(\theta\notin\frac{\pi}{2}\mathbb Z\), so the displayed composite is
not essentially unitary.
\end{proof}

\fi
\end{document}